\documentclass[11pt,article]{article}
\usepackage{amsfonts,amsbsy,amssymb,amsmath,amsthm,graphicx}  %

\usepackage[round]{natbib}
\usepackage{mathrsfs}

\usepackage{color}

\usepackage[colorlinks=true,linkcolor=red,citecolor=blue]{hyperref}

\def\1{1\kern-.20em {\rm l}}

\newtheorem{theorem}{Theorem}[section]

\newtheorem{lemma}[theorem]{Lemma}

\newtheorem{remark}[theorem]{Remark}

\def\1{1\kern-.20em {\rm l}}

\newcommand{\R}{\mathbb{R}}

\newcommand{\E}{\mathbb{E}}

\newcommand{\bx}{\mathbf{x}}
\newcommand{\by}{\mathbf{y}}

\numberwithin{equation}{section}

\begin{document}

\title{\bf Joint parametric specification checking of conditional mean and volatility in  time series models with martingale difference innovations}
\author{{ Kilani Ghoudi}$^{1,}$\footnote{Email: kghoudi@uaeu.ac.ae. Funding is provided in part by the United Arab Emirates University UPAR grant.} ~~
{ Na\^amane La\"ib}$^{2,}$\footnote{Email:naamane.laib@cyu.fr, naamane.laib@sorbonne-universite.fr } ~~ and~~{ Mohamed Chaouch}$^{3,}$\footnote{Email: mchaouch@qu.edu.qa}\\
$^1$College of Business and Economics, United Arab Emirates University.\\
$^2$
CY Cergy Paris Université, Laboratoire AGM, UMR 8088 du CNRS.\\
F-95000 Cergy, France.\\
$^3$Department of Mathematics, Statistics and Physics, Qatar University.
}

\maketitle
\begin{abstract}
Using cumulative residual processes, we propose joint goodness-of-fit  tests for  conditional means and  variances functions in the context of nonlinear  time series with martingale difference innovations. The main challenge comes from the fact  the cumulative residual process no longer admits, under the  null hypothesis,  a distribution-free limit. To obtain a practical solution one either transforms the process in order to achieve a distribution-free limit or approximates the non-distribution free limit using a numerical or a re-sampling technique. Here the three solutions will be considered.
 It is shown that the proposed tests  have nontrivial power against a class of root-n local alternatives, and  are suitable
 when the conditioning information set is infinite-dimensional, which allows  including  models like autoregressive conditional heteroscedastic stochastic  models with dependent innovations.
The approach presented assumes only certain conditions on the first- and second-order conditional moments, without imposing any
autoregression  model.
 The test procedures introduced are  compared
 with each other and with other competitors
 in terms of their power using a simulation study and a real data application. These simulations have shown that
 the   statistical powers of tests based on re-sampling or numerical approximation of the original statistics are in general slightly better than those based on a martingale transformation of the original process.

\vspace{1mm}

\end{abstract}

\section{Introduction}

A great deal of the data in econometrics, engineering and natural sciences occurs in the form of time series,
 where the data are naturally dependent and the volatility is, in general, function of the past.  One then  expects better forecast results if  additional information allowing the conditional variance to   depend upon the past realizations is supposed.  One of the earliest development in this area is the work of \cite{Engle} who introduced the so-called autoregressive conditional heteroscedastic (ARCH)  model, which
 has been extended in a number of directions. The most important of these is the extension designed to include moving-average parts, namely the generalized ARCH (GARCH) model introduced by  \cite{Bollerslev1986}.
These  models are often used to parameterize conditional
heteroscedasticity that appears in many financial time-series such as exchange rates and stock return.
In applications, GARCH models have been specified for data at different frequencies assuming that the rescaled innovations are independent and identically distributed  (i.i.d.).  \cite{Dros-Nijman1993} pointed out that, the common assumption in applications that rescaled innovations are independent is disputable, since it depends upon the available data frequency. They classified
 GARCH models  into three categories:
 Strong GARCH requires that rescaled innovations are independent,
 semi-strong GARCH assumes that rescaled innovations are
 martingale difference sequences,  the weak form, where  the
 martingale difference sequence assumption is relaxed.
The weak-GARCH representation  has been
used previously, for instance, by \cite{GWhite(2004)}, were the authors showed, through simulations, the effects of misspecification when the true model is a GARCH with an innovation term that follows an AR(1) process.
\cite{KRahbek(2005)} also noted the relevance of
this model.  \cite{Dahl.Iglesias(2007)}
provided an empirical application showing that this process has empirical relevance.

Diagnostic  tests are integral part of any modelling exercise.
Several  time series models  are given by specifying
	 	 conditional mean and conditional variance functions.   Testing the correct specification of these  quantities is of major importance in the model validation. A great deal of
tests proposed in the literature focuses on testing either the mean function or the volatility function for time series, but usually not both. \cite{escanciano2008, escanciano2010} discussed  joint tests for  parametric form of the mean and volatility functions. He also argued that if the mean is misspecified, tests of volatility functions are usually misleading.

 This paper derives   joint  goodness-of-fit tests for the conditional means  and variances functions  for    strictly  stationary ergodic time series  $\{X_i, i\ge 0\}$ with martingale difference  innovations $X_i-\mathbb{E}\{X_i|{\cal F}_{i-1}\}$,  where ${\cal F}_i$ is the
 $\sigma$-field generated by the observations
 obtained up to time $i$,  without imposing any type of  autoregressive model.  The martingale hypothesis is very important in economics theory, for instance,
 dynamic equilibrium approaches to macroeconomics have
 imposed martingale restrictions on numerous time series of interest (see, e.g.,  \cite{Durlauf} for more discussions on the martingale  hypothesis arising in other contexts of economic theory).
 The assumption of  martingale difference innovations considered here is  more general than the standard assumption of i.i.d innovations as it allows some dependence structure in the innovations. The framework we are considering here is suitable for cases in which the conditioning set is infinite-dimensional and may be
 used for models  that do not necessary satisfy  Markov property, particularly, semi-parametric models, where the  conditional mean and the conditional variance of $X_i$ given ${\cal F}_{i-1}$ have parametric forms. This includes most  processes  usually used for modeling financial time series with dependent innovations, such as  GARCH, ARMA-GARCH, exponential and threshold autoregressive processes with GARCH errors.

 Goodness-of-fit tests  for {\it parametric and semiparametric hypotheses of the regression function}  have been considered in the literature, with emphasis  on i.i.d innovations, see for instance \cite{Stute1997} who  presented non-parametric full-model checks for regression based on the limiting law of  the residual marked process, see also \cite{GC} for a  survey on the topic.  \cite{EPV} discussed a  general methodology for constructing nonparametric/semiparametric asymptotically distribution-free  tests  about regression models for possibly dependent data. Similar study investigated  the {\it autoregressive function in time series models}, see, e.g.,  \cite{DIE}, \cite{McZh}, \cite{La99}, \cite{KS}. 
 In the context of  {\it time series with martingale difference innovations},  \cite{SPGK2006} provided  non-parametric tests  based on residual cusums for testing the autoregressive function in higher-order time-series models, and \cite{EM} proposed data-driven smooth  asymptotically distribution-free tests based on the principal components of certain marked empirical processes for testing  the martingale difference hypothesis  of a possibly non-linear time series.\\
 Testing hypotheses about the {\it conditional variance function of regression models} are investigated by many  authors  in the literature.
 \cite{WaZh} considered a nonparametric diagnostic test for checking the constancy of the conditional variance function, in a nonparametric regression model, where the  co-variables are fixed design points, without  assuming a known  parametric form for the conditional mean function. \cite{DNV} proposed  a test procedure
 for testing the parametric form of the conditional variance in the common nonparametric regression model.
%
%
  \cite{KSong} discusses the problem of fitting a parametric model to the conditional variance function in a class of heteroscedastic regression models. Their   test is
   based on the supremum of the Khmaladze type martingale transformation of a certain partial sum process of calibrated squared residuals.
   The proposed statistical  test is shown to be consistent against a large class of fixed
   alternatives and to have nontrivial asymptotic power against a class of nonparametric local alternatives.
    Recently, \cite{PMA}  have proposed several  nonparametric statistical tests for  checking whether the conditional variances are equal in $k\geq 2$ location-scale regression models. Their
   procedure is based on the comparison of the error distributions under the null hypothesis of equality of variances functions.
\cite{PolYao} propose two  tests for  testing multivariate volatility functions using minimum volume sets and inverse regression. Their  tests are based on cumulative sums coupled with either minimum volume sets or inverse regression ideas.

Tests of conditional  {\it variance functions in time series context} were also previously considered in the literature. In particular,
\cite{AuTj},  focused on kernel estimate of the one step lagged conditional mean and variance functions for the purpose of identifying common linear models such as threshold and exponential autoregressive.   \cite{DIE}  established the consistency of regressogram type estimators of the  conditional  mean and conditional variance  functions. He deduced   nonparametric goodness-of-fit tests for known form of these functions.
\cite{CA} proposed a  Kolmogorov-Smirnov type statistic to test the homoscedasticity hypothesis,
  when the observations are assumed to be strongly mixing.   Their test uses only  a subsample which induces a  loss of information and power.
   \cite{N}  presented a procedure, based on marked empirical process,  for testing the goodness-of-fit of the conditional variance function of a Markov model of order one.

For time {\it series with  martingale difference innovations}
 \cite{LaChe} considered a class of nonlinear semi-parametric models and established the local asymptotic normality for cumulative residual process.
 They also derived an efficient  simultaneous   test  for testing the conditional mean and
 the conditional variance functions. \cite{LaLou} provided a  test  of  conditional  homoscedasticity hypothesis  of  the  one-step  forecast  error  in the context of  first-order AR-ARCH model. Their works was extended by \cite{La2003}  for the context of time series with martingale difference innovations. The author established the asymptotic of the cumulative residual process and developed a test  for homoscedasticity   when the innovations are independent of the past of  ${\cal F}_{i-1}$. \cite{chen-etal2015}
 developed  two tests for parametric volatility function of a diffusion model, with i.i.d. innovations, based on Khamaladze's  martingale transformation. Their tests use the structural properties of the diffusion process and do not require the estimation or the specification of  the drift function.
\cite{escanciano2008}  proposed a class of joint and marginal spectral diagnostic tests for parametric conditional means and variances  of  time series models. The proposed tests are not distribution-free and the author introduced a   bootstrap  procedure that should be used to  implement    these tests.
 \cite{escanciano2010} constructed    asymptotic distribution-free joint  specification tests, that can be  applied in many  financial and economic time series including   GARCH and  ARMA-GARCH models. These  tests are  based on carefully weighted residual empirical process. The weights are chosen in a way to insure that  the weighted empirical process of residuals admits a distribution free limit.
 It is shown that the proposed tests generalize those of  \cite{Wooldridge1990}. Note however that,  the performance of  the constructed tests strongly depend on the choice of the weights.

 In this paper we  develop a joint test for parametric form specification of the conditional mean and variance functions when innovations satisfy the martingale difference hypothesis and are allowed to depend on the past $\sigma$-field.
The main challenge comes from the fact  the marked cumulative residual process no longer admits a distribution-free limit. To obtain a practical solution one either transforms the process in order to achieve a distribution-free limit or approximates the non-distribution free limit using numerical or re-sampling techniques. Here, three solutions, based on Khmaladze transform of cumulative residual process, a multiplier bootstrap re-sampling procedure, and a numerical approximation of the limiting distribution, are considered and their finite sample performance is compared with \cite{escanciano2010}'s  procedure.

Though Khmaladze transform of the cumulative residual process was considered in previous works, such as \cite{chen-etal2015} or \cite{stute-etal1998}, the setting addressed here is more general and does not assume any explicit data-generating model or independence of innovations. It may be applied for several non-linear time series models with martingale innovations. The second added value of this paper is that it compares numerically the three classical procedures used for constructing joint tests based on the cumulative residual process. The numerical study particularly revealed that the Khmaladze transform is slightly less performant than re-sampling and numerical approximation techniques.

The rest of the paper is organize as follows: Section \ref{mr} defines the problem and states some preliminaries results. Section \ref{sec3} defines and establishes the asymptotic behavior of the martingale transformation applied to the cumulative residual process. Section \ref{sec4}  introduces the marginals as well as the joint parametric specification Cram\'er von-Mises-type tests statistics for the conditional mean and conditional variance functions.  Subsection \ref{MTA} presents the martingale transform based test statistics. A numerical approximation procedure for the asymptotic distribution of the test statistics based on the original process is given in Subsection \ref{secnumapp}.  A re-sampling algorithm for test statistics based on the original cumulative residual process is also detailed in  Subsection \ref{secmult}.
A comparison between  these   statistical procedures,  via  simulations,  is outlined in Section \ref{sim} and an application to real data is given in Section \ref{app}. A conclusion summarizing our findings is given in Section \ref{conc}.  All proofs are provided in the Appendix.

\section{Assumptions and Main results}\label{mr}
Let $\{(X_i,Z_i), i\in  \mathbb{Z}\}$ be a strictly
stationary   ergodic   defined on the probability space $(\Omega, {\cal A}, \mathbb{P})$. The random variables $X_i$'s are real-valued with common
continuous distribution function ${F}$.
For each $i\ge 1$ we let $I_{i-1}=(X_{i-1},X_{i-2},\ldots, Z_{i})$ denote the past information at time $i$. We let ${\cal F}_i=\sigma(I_0, I_1, \ldots, I_i)$ denotes
$\sigma$-field generated by $I_0, \ldots, I_i$. The purpose is to verify if the  conditional mean $\mu(I_{i-1})=\mathbb{E}\left(X_i \vert {\cal F}_{i-1} \right)$  almost surely (a.s.) and the conditional variance $\sigma^2(I_{i-1})=\mathrm{Var}\left( X_i \vert {\cal F}_{i-1} \right)$ a.s.  of
$X_i$ given ${\cal F}_{i-1}$ satisfy, respectively, the following relations
$
\mu(I_{i-1})=m(\theta_0,I_{i-1})$ a.s. and
$\sigma^2(I_{i-1})= \sigma^2(\theta_0,I_{i-1})$ a.s., where $m(\theta,\cdot)$ and $\sigma^2(\theta,\cdot)$ are $\mathcal{F}_{i-1}$-measurable known parametric functions that depend on a finite dimensional vector of parameters $\theta\in \mathbb{R}^d$ for $d\ge 1$, and assumed to be finite with probability one.
  More precisely, we are interested in examining   hypotheses stating that
 \begin{eqnarray}
 & \ (H_0)& \mu(I_{i-1})=m(\theta_0,I_{i-1}) \ \mbox{a.s. and }\sigma^2(I_{i-1})=\sigma^2(\theta_0,I_{i-1}) \ \mbox{a.s.} \
 \mbox{ versus } \nonumber\\
 &(H_1)& \quad \mu(I_{i-1})\neq m(\theta, I_{i-1}) \mbox{ or } \sigma^2(I_{i-1})\neq \sigma^2(\theta_0, I_{i-1}).
 \label{hypotheses2}
 \end{eqnarray}
We assume throughout this manuscript that for the true value of $\theta$, denoted $\theta_0$ belongs to the interior of some compact subset $\mathbf{\Xi}\subset\mathbb{R}^d$. We also assume that the sequence of innovations $\{ X_i -m(\theta_0,I_{i-1}):\, \ i \geq 0 \}$ is a sequence of martingale differences with respect to ${\cal F}_{i}$,  that  is $X_i -m({\theta_0},I_{i-1})$ is ${\cal F}_{i}$-measurable and
$\mathbb{E}\left[\left(X_i -m({\theta_0},I_{i-1}) \right)\ |    {\mathcal {F}}_{i-1}\right]=0 \ \ \text{a.s.}$
This condition combined with  the fact that the set of information $I_i$ has an infinite dimensional will  allow to consider  non markovian processes such as the  ARMA-GARCH  process and all others.
In practice the set $I_i$ is  	not observable and may be estimated, see Remark \ref{Remark1} and  \cite{escanciano2010} for more discussion.

To test the specification of the conditional mean function we will use $W^1_{\theta}(  X, I):=X-m({\theta}, I)$ and to test
for the conditional variance function we will use   $W^2_{\theta}(  X, I):=(X-m({\theta},I))^2-\sigma^2({\theta}, I)$.
Following  \cite{escanciano2007}, \cite{escanciano2008}, \cite{KS},  \cite{La2003} and  \cite{N} we  introduce  the following cumulative  empirical residuals processes
 \begin{eqnarray*}
 \mathbb{D}_n^k(x)&=&n^{-1/2}\sum_{i=1}^n W^k_{\theta_0}(  X_i, I_{i-1})
\1{\{X_{i-1}\leq x\}}, \ \
 x\in \mathbb{R} \mbox{ and } k=1,2. \\
 \widehat{\mathbb{D}}_n^k(x)&=&n^{-1/2}\sum_{i=1}^n  W^k_{\theta_n}(  X_i, I_{i-1})\1{\{X_{i-1}\leq x\}}, \ \ x\in
 \mathbb{R},\mbox{ and } k=1,2\\
 \end{eqnarray*}
 where $\1(A)$ is the indicator function of the set $A$, $\theta_n$ is
 consistent estimators of $\theta$.  Note that,  under $H_0$,   $\mathbb{E}\{W^k_{\theta_0}(  X_i, I_{i-1})|\mathcal{F}_{i-1}\}=0$ for $k=1,2$.  To test  specification of the conditional mean and variance jointly,
 we introduce the bivariate processes $\mathbb{D}_n$ and $\widehat{\mathbb{D}}_n$ defined  by $\mathbb{D}_n(x):=(\mathbb{D}_n^1(x),\mathbb{D}_n^2(x))^\top$ and $\widehat{\mathbb{D}}_n(x):=(\widehat{\mathbb{D}}_n^1(x),\widehat{\mathbb{D}}_n^2(x))^\top$, where the script $B^\top$ stands  for the transpose of the  matrix $B$.
Note that letting $W_\theta(X,I):=\{W^1_\theta(X,I),W^2_\theta(X,I)\}^\top$ one sees that
$$\mathbb{D}_n(x)= n^{-1/2}\sum_{i=1}^n W_{\theta_0}(  X_i, I_{i-1})
\1{\{X_{i-1}\leq x\}}$$ and
$$\widehat{\mathbb{D}}_n(x)= n^{-1/2}\sum_{i=1}^n W_{\theta_n}(  X_i, I_{i-1})
\1{\{X_{i-1}\leq x\}}.$$

\begin{remark}
 Comparing  the empirical process $\widehat{\mathbb{D}}_n^k$ with that used  in \cite{escanciano2010}, one notice that
$\widehat{\mathbb{D}}_n^k$ is  a marked empirical process based on $X_{i-1}$ with marks/weights given by $W_{\theta_n}^k$ while the process used in \cite{escanciano2010} is based on the residuals $\widehat{\epsilon}_i$, with weights carefully estimated from data.

\end{remark}
\begin{remark}\label{Remark1}
As pointed in \cite{escanciano2010}, the past information $I_{i-1}$ is not completely observable and needs to be estimated by $\hat{I}_{i-1}$. Such estimation usually involves  replacing the unknown initial state $I_0$ by some quantity $\hat I_0$. One must insure that replacing $I_{i-1}$ by $\hat{I}_{i-1}$ does not affect the asymptotic behaviour of the process $\widehat{\mathbb{D}}_n$. One easily verifies that if for $k=1,2$, $$n^{-1/2}\sum_{i=1}^n \mathbb{E}\left\{\sup_{\theta\in \mathbf{\Xi}}|W^k_{\theta}(  X_i, \hat{I}_{i-1})-W^k_{\theta}(  X_i, I_{i-1})|\right\}=o(1)$$
then the asymptotic behaviour of $\widehat{\mathbb{D}}_n$ is not altered when $I_{i-1}$ is replaced by  $\hat{I}_{i-1}$. We will assume that such condition holds and will use $I_{i-1}$ in the definition of
$\widehat{\mathbb{D}}_n$ throughout the manuscript. The discussion in \cite{escanciano2010} shows that this condition holds true for ARMA-GARCH models in particular.
\end{remark}

 The limiting laws of $\mathbb{D}_n^1$ and  $\widehat{\mathbb{D}}_n^1$ are given in \cite{KS} and \cite{escanciano2007}. The limiting behavior of $\mathbb{D}_n^2$ is given in  \cite{La2003}, who also obtained the limit law of $\widehat{\mathbb{D}}_n^2$ in the special case $\sigma^2(\cdot)=\zeta^2$ with $\zeta^2\in (0,\infty)$. The next section studies the asymptotic for the bivariate process $\widehat{\mathbb{D}}_n$ .
\subsection{Limiting law of the process $\widehat{\mathbb{D}}_n$ under the null hypothesis}

The following notations are used in the rest of the paper.  $\|v \|$ denotes the Euclidian norm of the vector $v$ and for any matrix $A$, $\|A\|=\sup_{v:\; \|v\|=1}\{\|Av\|\}$ is the associated matrix norm for the matrix $A$. For any bounded  function $f$, let $\Vert f \Vert =\sup_x\left | f(x)\right|$ and for any
$(x, y)\in \mathbb{R}^2$, set
$x\wedge y= \min (x, y)$ and  $x\vee y =\max(x, y)$.

The asymptotic results are stated under  the following  assumptions:\\
{\bf Assumption A1:} For each $k=1,2$,
\begin{enumerate}
\item $\mathbb{E} \left(\vert W^k_{\theta_0}( X_i ,I_{i-1})\vert^{2}\right)<\infty $.
\item $\lim_{n\to \infty} \mathbb{E}\left\{ W^k_{\theta_0}( X_i ,I_{i-1})\1\{|W^k_{\theta_0}( X_i ,I_{i-1})|>\delta \sqrt{n}\}\right\}=0 $  for any real $\delta>0$.
\item $K_k(x)=\lim_{n\to\infty}\frac{1}{n}\sum_{i=1}^n\mathbb{E} \left( W^k_{\theta_0}( X_i ,I_{i-1})^2\1\{X_{i-1}\le x\}\right)$ is non-decreasing continuous function of $x$.
\item $ \mathbb{E} \left( W^k_{\theta_0}( X_i ,I_{i-1})^2\1\{x\le X_{i-1}\le y\}|\mathcal{F}_{i-2}\right)=C_{k,i}|K_k(y)-K_k(x)|$ with $C_{k,i}$ such that
$\frac{1}{n}\sum_{i=1}^n \mathbb{E}|C_{k,i}|=O(1)$.
\end{enumerate}
{\bf Assumption A2:}
Under $(H_0)$, the estimators $\theta_n$ of $\theta_0$ is  such that
 $$n^{1/2}(\theta_n -\theta_0)=    n^{-1/2}\sum_{i=1}^n\phi^\star(X_i, I_{i-1}, \theta_0)+ o_{\mathbb{P}}(1)$$
where $\phi^\star$ is an $\mathbb{R}^d$-valued  measurable function  satisfying $\mathbb{E}\left( \phi^\star(X_i,I_{i-1},\theta_0) \vert {\cal F}_{i-1}\right)=0$ a.s.  and
 $$\Sigma_0=\lim_{n\to\infty}\frac{1}{n}\sum_{i=1}^n\mathbb{E}\left( \phi^\star(X_i,I_{i-1},\theta_0) \phi^\star(X_i,I_{i-1},\theta_0)^\top \right)$$
 exists and is positive definite.\\
{\bf Assumption A3:}
For $k=1,2,$ let  $\dot{W}^k_{\theta^\star}(X_i,I_{i-1})$ denote the gradient of $W^k_\theta(X_i,I_{i-1})$ with respect to the components of $\theta$ evaluated at a fixed point $\theta^\star$. Assume that $\mathbb{E}\|\dot{W}^k_{\theta_0}(X_i,I_{i-1}))\|\le C<\infty$ and $$\left\|W^k_{\theta}(X_i,I_{i-1})-W^k_{\theta_0}(X_i,I_{i-1})-(\theta-\theta_0)^\top \dot{W}^k_{\theta_0}(X_i,I_{i-1})\right\| \le \| \theta-\theta_0\| M_1(X_i,I_{i-1})\lambda_1(\| \theta-\theta_0\|),$$
  where $M_1$ and $\lambda_1$ are positive functions satisfying $\mathbb{E}(M_1(X_i,I_{i-1})\le C<\infty$ and $\lambda_1(t)$ goes to zero as $t\to 0$.\\

Note that Assumption (A1) is an adaptation of conditions (A--D) of Theorem 1 of \cite{escanciano2007b} to the context of stationary and ergodic sequence.
As argued in \cite{escanciano2007b} these are among the weakest conditions to ensure the weak convergence of marked empirical processes. Assumption (A2) and (A3) are commonly used to ensure the convergence of $\sqrt{n}(\theta_n-\theta_0)$ and to validate the expansion of the process $\widehat{\mathbb{D}}_n$. These assumptions hold for most commonly used models and estimation procedures, see, for example \cite{KS, escanciano2007, escanciano2010} for discussion and details.


From now on let
$\mathbf{D}(\overline{\mathbb{R}})$ denotes the space of c\`adl\`ag functions and for $x,y\in \mathbb{R}$ define
 \begin{eqnarray}\label{def6*}
 K(x,y)=\rm{Cov}(\mathbb{D}_n(\bx),\mathbb{D}_n(\by))=\left(\begin{array}{ll}K_{1}(x\wedge y)&K_{1,2}(x\wedge y)\\
 K_{1,2}(x\wedge y)&K_{2}(x\wedge y)\end{array}\right),
\end{eqnarray}
where $K_1$ and $K_2$ are given in Assumption (A1) and for $x\in \mathbb{R}$
\begin{eqnarray}\label{def6}
K_{1,2}(x)&=& \lim_{n\to \infty}\frac{1}{n}\sum_{i=1}^n\mathbb{E}\left[W^1_{\theta_0}(X_i,I_{i-1})W^2_{\theta_0}(X_i,I_{i-1})  \1\left(  X_{i-1}\leq
x\right)\right].
\end{eqnarray}
Let $\mathbb{D}$ denotes the Gaussian process  with covariance function $K$ defined above.
The following result summarizes the weak convergence of the processes ${\mathbb{D}}_n$ and $\widehat{\mathbb{D}}_n$.
\begin{theorem}\label{TheoremA} If Assumption (A1) holds then, under $H_0$,  ${\mathbb{D}}_n$ converges weakly to $ {{\mathbb{D}}}.$
If Assumptions (A1)-(A3) hold true, then, under $H_0$,
$\widehat{\mathbb{D}}_n$ converges weakly to $ {\widehat{\mathbb{D}}},$
where  ${\widehat{\mathbb{D}}}$ is
a   centered Gaussian process  given by
\begin{eqnarray}\label{def7}
\widehat{\mathbb{D}}(x)=\mathbb{D}(x)-\Gamma_{\theta_0}^\top(x) \Theta
\end{eqnarray}
where  $\Theta$ is a multivariate normal with mean zero and covariance $\Sigma_0$ and the function 
$\Gamma_{\theta_0}(x)=(\Gamma^1_{\theta_0}(x),\Gamma^2_{\theta_0}(x))$ with
$$\Gamma^1_{\theta_0}(x)=\lim_{n\to \infty}\frac{1}{n}\sum_{i=1}^n \mathbb{E}\left[\dot{m}(\theta_0,I_{i-1})\1\{X_{i-1}\le x\}\right]$$ and
$$\Gamma^2_{\theta_0}(x)=\lim_{n\to \infty}\frac{1}{n}\sum_{i=1}^n \mathbb{E}\left\{\dot{\sigma}^2( \theta_0, I_{i-1})\1\{X_{i-1}\le x\}\right\}.$$
The covariance function of $\widehat{\mathbb{D}}$ is given by
\begin{eqnarray}\label{def9}
{\mathcal{K}}(x, y)= {K}(x,y)-\Gamma^\top_{\theta_0}(x) G(y)-G^\top(x)\Gamma_{\theta_0}(y)
+ \Gamma^\top_{\theta_0}(x)\Sigma_0\Gamma_{\theta_0}(y),
\end{eqnarray}
where $G(x)=\mathrm{Cov}(\Theta,\mathbb{D})=(G^1(x),G^2(x))$ with
$$G^k(x)=\lim_{n\to\infty}\frac{1}{n}\sum_{i=1}^n\mathbb{E}\left[ W^k_{\theta_0}(X_i,I_{i-1})\phi^\star(X_i,I_{i-1},\theta_0) \1\left(  X_{i-1}\leq x\right)\right] \mbox{ for } k=1,2.$$
\end{theorem}
\begin{remark}
The limiting covariance function ${\mathcal{K}}$ is in general complicated, therefore classical statistics based on the process $\widehat{\mathbb{D}}_n$ do not admit distribution-free limits. To overcome this, one usually uses one of the following approaches. The first approach consists in transforming the process $\widehat{\mathbb{D}}_n$ in a such way to achieve a distribution-free limit and then using the transformed process to  define test statistics. While the second approach uses $\widehat{\mathbb{D}}_n$ to construct test statistics and then adopts either a numerical approximation or a re-sampling technique to estimate the non-distribution free limit. Both approaches will be discussed and compared in the rest of this manuscript. The transformation is discussed in Section \ref{sec3} while the numerical and re-sampling approximations are outlined in Section \ref{sec4}.
\end{remark}


\subsection{Limiting law of the process $\widehat{\mathbb{D}}_n$ under local alternatives}\label{loc}
In this section we establish the limiting behavior of the process $\widehat{\mathbb{D}}_n$ under local alternatives $H_A$ defined as follows.
$$H_A:\;  \mu(I_{i-1})=m(\theta_0,I_{i-1}) +a_1(I_{i-1})/\sqrt{n}\mbox{ and }\sigma^2(I_{i-1})=\sigma^2(\theta_0,I_{i-1}) +a_2(I_{i-1})/\sqrt{n},$$
where  $a_1$ and $a_2$ are some   non-zero functions indicating the direction of departure from the null hypothesis. We assume that	$\mathbb{E}|a_k(I_{i-1})|<\infty$ for $k=1,2$.
The following assumption is needed to establish the limiting behavior under $H_A$.\\
\noindent{\bf Assumption (L1)} Under $H_A$, the estimator $\theta_n$ satisfies
$$n^{1/2}(\theta_n -\theta_0)=    n^{-1/2}\sum_{i=1}^n\phi^\star(X_i, I_{i-1}, \theta_0)+ \xi_A+o_{\mathbb{P}}(1),$$ where $\xi_A\in\mathbb{R}^d$ is a random/nonrandom vector  and
$\phi^\star$ is as in Assumption (A2).

The following Theorem shows that our tests are able to detect local alternatives of the type described by $H_A$.
\begin{theorem}\label{localt}
	Under $H_A$, if Assumptions A1, A3 and L1 hold, the process $\widehat{\mathbb{D}}_n$ converges weakly to $\widehat{\mathbb{D}} +\Psi_A(x) -\Gamma^\top_{\theta_0}(x) \xi_A,$ where $\Psi_A(x)=\mathbb{E}\left\{\left(\begin{array}{l}a_1(I_{i-1})\\ a_2(I_{i-1})\end{array}\right)\1\{X_{i-1}\le x\}\right\}$.
\end{theorem}

\section{Khmaladze transform of the process $\widehat{\mathbb{D}}_n$ }\label{sec3}
This section presents a martingale transformation of the Khmaladze type that will be applied to the process  $\widehat{\mathbb{D}}_n$ in order to achieve a distribution-free limit for classical test statistics. It worth mentioning that the Khmaladze transformation will ensure that the limit of the transformed process is Brownian motion.

To define such transformation, assume the functions  $\Gamma_{\theta_0}^k(x)$, for $k=1,2$, can be written as
\begin{equation}\label{Kh1}
   \Gamma_{\theta_0}^k(x)=\int_{-\infty}^x g_k(t)K_{k}(dt),
 \end{equation}
 where $g_k(t)=(g_{1,k}(t),\ldots, g_{d,k}(t))^\top$.
Let the matrix $A_k$ for $k=1,2$ be defined as $A_k(x)=\int_{x}^\infty g_k(u)g_k^\top(u){K_{k}}(du)$. Assume that there exists $x_0<\infty$ such $A_k(x_0)$ is non-singular. Note that it follows from the definition of $A_k(x)$ that $A_k(x)-A_k(x_0)$ is non-negative definite for all $x\le x_0$ implying that $A_k(x)$ is invertible for all $x\le x_0$ whenever $A_k(x_0)$ is invertible.

Following the concept in \cite{Khma}, for any function $\mathbf{f}$ consider the linear transformations $T^k(\mathbf{f})$ for $k=1,2$  given  by
\begin{eqnarray}\label{Trans.1}
(T^k{\mathbf{f}})(x)= {\mathbf{f}}(x)-\int_{-\infty}^x (g_k (y))^\top A_k^{-1} (y)\left[ \int_y^\infty g_k(z)  {\mathbf{f}}(dz)    \right] {K_{k}}(dy),
\end{eqnarray}
for all $x\le x_0$.
The Khmaladze transform of the cumulative residual process considered here is an application of the transformation $T^1$ and $T^2$ to the processes $\widehat{\mathbb{D}}^1_n$ and $\widehat{\mathbb{D}}^2_n$. It is formally defined  by

\begin{multline}\label{trsns.Dn.1}
T(\widehat{\mathbb{D}}_n(x)) =  \left( \begin{array}{l} T^1(\widehat{\mathbb{D}}^1_n(x))\\
T^2(\widehat{\mathbb{D}}^2_n(x))\end{array}\right)\\
=\left( \begin{array}{l}\widehat{\mathbb{D}}_n^1(x)-\int_{-\infty}^{x} (g_1(y))^\top A_1^{-1} (y)\left[ \int_y^\infty g_1(z) \widehat{\mathbb{D}}_n^1(dz)    \right] {K_{1}}(dy) \\
\widehat{\mathbb{D}}_n^2(x)-\int_{-\infty}^x (g_2(y))^\top A_2^{-1} (y)\left[ \int_y^\infty g_2(z) \widehat{\mathbb{D}}_n^2(dz)    \right] {K_{2}}(dy)
\end{array}\right).
\end{multline}
Easy manipulations show that for $k=1,2$
\begin{multline}\label{trsns.Dn.1.2}
T^k(\widehat{\mathbb{D}}^k_n(x))=\frac{1}{\sqrt{n}}\sum_{i=1}^n W^k_{\theta_n}(X_i,I_{i-1})
 \biggl[ \1{\left\{X_{i-1} \leq x \right\}} \\
 - \int_{-\infty}^{\min(x, X_{i-1})} (g_k(y))^\top  A_k^{-1} (y)
{K_{k}}(dy)  g_k(X_{i-1}) \biggr].
\end{multline}
  Observe also that $T^k(\widehat{\mathbb{D}}^k_n)(x)$ cannot be used in practice to build test statistics since it depends on several unknown quantities, namely,  $g_k$ and $K_{k}$. Replacing  $g_k$ and $K_{k}$ by their consistent estimates $\hat g_k$ and  $\hat{K}_{k}$ respectively, one obtains an empirical version of $T^k(\widehat{\mathbb{D}}^k_n)$ defined as follows
\begin{multline*}
T^k_n(\widehat{\mathbb{D}}^k_n(x))=\frac{1}{\sqrt{n}}\sum_{i=1}^n W^k_{\theta_n}(X_i,I_{i-1})
 \biggl[ \1{\left\{X_{i-1} \leq x \right\}}\\
  - \int_{-\infty}^{\min(x, X_{i-1})} (\hat g_k(y))^\top  \hat A_{k}^{-1} (y)
{\hat{K}_{k}}(dy)  \hat g_k(X_{i-1}) \biggr],
\end{multline*}
where $\hat A_k(x)=\int_{x}^\infty \hat g_k(u)\hat g_k^\top(u){\hat K_{k}}(du)$.
The Following  assumptions are needed to ensure the convergence of $T_n(\widehat{\mathbb{D}}_n)=(T^1_n(\widehat{\mathbb{D}}^1_n),T^2_n(\widehat{\mathbb{D}}^2_n))^\top$.\\

\noindent {\bf Assumption K1:} Assume there exists $x_0<\infty$  such that $A_k(x_0)$ is invertible for $k=1,2$.\\
\noindent {\bf Assumption K2:} Assume that $\mathbb{E}\|W_{\theta_0}^k(X_i,I_{i-1})g_k(X_{i-1})\|\}$, $\mathbb{E}\{M_1(X_i,I_{i-1})\|g_k(X_{i-1})\|\}$, $\mathbb{E}\{M_2(X_0)\|g(X_0)\|\}$, $\mathbb{E}\{ \|g(X)\|\}$, $\mathbb{E}\{ \|g(X)g(X)^\top\|W_{\theta_0}^k(X_i,I_{i-1})^2\}$ and $\mathbb{E}\{ \|g(X)g(X)^\top\|\}$ are all finite.\\
\noindent {\bf Assumption K3:} Assume that $\|\hat{g}_k-g_k\|$ converges to zero in probability,
and that $\sup_{y\in \R}\left\|\int_y^\infty (g_k(z)-\hat g_k(z)){\mathbb{D}}_n(dz)\right\|$ converges to zero in probability.\\
\noindent {\bf Assumption K4:} For each $k=1,2$ and for each $j=1,\ldots, d$,
\begin{enumerate}
\item $\mathbb{E} \left(| W^k_{\theta_0}( X_i ,I_{i-1})g_{j,k}(X_{i-1})|^{2}\right)<\infty $.
\item $\lim_{n\to \infty} \mathbb{E}\left\{ | W^k_{\theta_0}( X_i ,I_{i-1})g_{j,k}(X_{i-1})|^{2}\1\{|W^k_{\theta_0}( X_i ,I_{i-1})g_{j,k}(X_{i-1})|>\delta \sqrt{n}\}\right\}=0 $, for any $\delta>0$.
\item $\bar K_{j,k}(x)=\mathbb{E} \left( W^k_{\theta_0}( X_i ,I_{i-1})^2g_{j,k}(X_{i-1})^2\1\{X_{i-1}\le x\}\right)$ is non-decreasing continuous function of $x$.
\item $ \mathbb{E} \left( W^k_{\theta_0}( X_i ,I_{i-1})^2 g_{j,k}(X_{i-1})^2 \1\{x\le X_{i-1}\le y\}|\mathcal{F}_{i-2}\right)=C^j_{k,i}|\bar K_{j,k}(y)-\bar K_{j,k}(x)|$ with $C^j_{k,i}$ such that
$\frac{1}{n}\sum_{i=1}^n \mathbb{E}|C^j_{k,i}|=O(1)$.
\end{enumerate}

The next result establishes the weak convergence of $T_n(\widehat{\mathbb{D}}_n)$.
\begin{theorem}\label{Tn}
If assumptions A1--A4 and K1--K4 are satisfied  then for $k=1,2$
 $$T_n^k(\widehat{\mathbb{D}}_n^k) \rightsquigarrow {\mathbb{D}^k} \quad\text{in}\quad \mathbf{D}((-\infty, x_0]),$$
where $\mathbb{D}^k = \mathbb{W}_k\circ K_{k}$ is the $k$-th component in the process $\mathbb{D}$ defined in Theorem \ref{TheoremA}.
\end{theorem}
\begin{remark}
Note that the process $\mathbb{D}$ has marginal $\mathbb{W}_1\circ K_{1}$ and $\mathbb{W}_2\circ K_{2}$ where $\mathbb{W}_1$ and $\mathbb{W}_2$ are Brownian motions. In general $\mathbb{W}_1$ and $\mathbb{W}_2$ are not independent. The distribution of $\int \| \mathbb{D}\|^2$ will only be distribution free if $\mathbb{W}_1$ and $\mathbb{W}_2$  are independent. This defeats the purpose of the martingale transformation  in general. Though the above results clearly show that the transformation provides a distribution-free limit for testing the mean or the variance function separately. The transformation is only useful for joint testing of conditional mean and variance functions only in the case of independent components of the bivariate process $\mathbb{D}$. This was  the case in the work of  \cite{chen-etal2015}. In general, the components $\mathbb{D}^1$ and $\mathbb{D}^2$  will be independent if  $K_{12}=0$ which translates to a condition on the third moment of $\varepsilon_i=X_i-m(\theta,I_{i-1})$. In particular,  if  $\mathbb{E}(\varepsilon_i^3|\mathcal{F}_{i-1})=0$ a.s. then the components $\mathbb{D}^1$ and $\mathbb{D}^2$  are independent. This holds true for all time series models that assume that the errors are i.i.d provided that we add a condition stating that the error is symmetric or just has zero third conditional moment.

\end{remark}
\begin{remark}
Note that the conditions imposed on the estimate $g_n$ are similar to those used in \cite{Ba} and \cite{BaNG}. Note also that the {$\|\hat{g}_k-g\|$} converges to zero in probability can be weakened to  {$\int \|\hat{g}_k(x)-g(x)\|^2 F(dx)$} converges to zero in probability. But for simplicity, the proof is presented with the stronger assumption.  The second condition on $\hat{g}_k$, in Assumption (K3), is usually verified using the structure of the estimator and the nature of the process $\mathbb{D}_n$, for details see the discussion in \cite{Ba}.
\end{remark}

\section{Test statistics}\label{sec4}
Test statistics for the hypothesis $(H_0)$ are obtained by considering continuous functional  $\mathcal{G}$ of the process  $\widehat{\mathbb{D}}_n$ or of its
Khmaladze transform $T_n(\widehat{\mathbb{D}}_n)$. That is, given a continuous functional $\mathcal{G}$ a test statistic can be obtained  using $\mathcal{G}\left(\widehat{\mathbb{D}}_n\right)$ or $\mathcal{G}\left(T_n(\widehat{\mathbb{D}}_n)\right)$. The continuous mapping theorem provides the asymptotic behavior of $\mathcal{G}\left(\widehat{\mathbb{D}}_n\right)$ and $\mathcal{G}\left(T_n(\widehat{\mathbb{D}}_n)\right)$.
The most commonly used functional, in practice, are of the  Cram\'er-von-Mises or Kolmogorov-Smirnov type. For simplicity, the rest of this paper only focuses on Cram\'er-von-Mises type statistics, but the approaches discussed apply to any continuous functional. To be specific, the following {marginal} Cram\'er-von-Mises test statistics
\begin{eqnarray}\label{SS1}
S_{n}^k &=& \int_{-\infty}^{\infty} \left[\widehat{\mathbb{D}}_n^k(x)\right]^2 F_n(dx)\\ \nonumber
&=&\frac{1}{n}\sum_{i=1}^n\sum_{j=1}^n \left\{1 - F_n(X_{i-1}\vee X_{j-1})\right\}
W^k_{\theta_n}(X_i,I_{i-1})W^k_{\theta_n}(X_j,I_{j-1})
\end{eqnarray}
for $k=1,2$, will be used here. Clearly, the statistics $\{S_{n}^k\}_{k=1,2}$  given in \eqref{SS1} converge in distribution to $S^k:=\int_{-\infty}^{\infty} [\widehat{\mathbb{D}}^k(x)]^2 F(dx)$ which is equivalent to $\int_{0}^{1}[\widehat{\mathbb{D}}^k(F^{-1}(u))]^2 du$.

The statistic ${S}_{n}^1$ can be used to test the conditional mean function and the statistic ${S}_{n}^2$ can be used to test the conditional variance function. To conduct a joint test of mean and variance functions one must combine ${S}_{n}^1$ and ${S}_{n}^2$. To combine both statistics, \cite{chen-etal2015} used ${S}_{n}^1+{S}_{n}^2$ and $\max({S}_{n}^1,{S}_{n}^2)$. Such combinations were adequate since  in his case   both statistics had the same limiting distribution. In our context, this is not true in general, hence to combine these test statistics we will need a more general  approach such as the ones discussed in  \cite{Ghoudi/Kulperger/Remillard:2001} and \cite{Genest/Remillard:2004}. To be specific,  define
$S_{n}^\star=S_{n}^1/L_1+S_{n}^2/L_2$, $S_{n}^\bullet=-2[\ln(P_v((S_{n}^1))+\ln(P_v(S_{n}^2))]$ and $S_{n}^\circ=\max(S_{n}^1/L_1,S_{n}^2/L_2)$, where
$L_k=\mathbb{E}(S^k)=\int K_k(x)F(dx)$ and $P_v (S^k_n)$ denotes the P-value of the statistics $S^k_n$.
According to \citet[ pp.~99-101]{Fisher:1950}, when $S^1$ and $S^2$ are independent $S_{n}^\bullet$ converges to Chi-square distribution with $4$ degrees of freedom and provides the optimal way of combining $S^1_n$ and $S^2_n$. It will be seen, in the simulation and application sections, that $S_{n}^\star$ and $S_{n}^\circ$ are easier to compute and  have similar power to $S_{n}^\bullet$.

\subsection{{Martingale transform based test statistics}}\label{MTA}
Next, we repeat a similar approach using the transformed process $T_n(\widehat{\mathbb{D}}_n)$ yielding
\begin{eqnarray}\label{CVMKhT}
	\widetilde{S}_{n}^k =\int_{-\infty}^\infty \frac{\left(T_n^k(\widehat{\mathbb{D}}_n^k)(x) \right)^2}{(\gamma^k_n)^2} {\hat{K}_{k}}(dx) = \frac{1}{n (\gamma^k_n)^2} \sum_{i=1}^n \left(T_n^k(\widehat{\mathbb{D}}_n^k)(X_{i-1}) \right)^2 W^k_{\theta_n}(X_i,I_{i-1})^2,
\end{eqnarray}
where
\begin{eqnarray}\label{Kkkn}
\hat{K}_{k}(x)&=& \lim_{n\to \infty}\frac{1}{n}\sum_{i=1}^n W^k_{\theta_n}(X_i,I_{i-1})^2  \1\left(  X_{i-1}\leq
x\right), \ x\in \mathbb{R} \mbox{ and } k=1,2.
\end{eqnarray}
and $\gamma_n^k=\lim_{x \to \infty}\hat{K}_{k}(x)=\frac{1}{n}\sum_{i=1}^n W^k_{\theta_n}(X_i,I_{i-1})^2$.
Observe that Theorem \ref{Tn} and the continuous mapping Theorem imply that, under the null hypothesis, $\tilde S_{n}^k$ converges in distribution to $\tilde{S}^k:=\int_0^1 {\mathbb{W}_k(u)^2} du$, where ${\mathbb{W}_k}$ is a standard Brownian motion. According to \cite{SW} (page 748), for $\widetilde{S}_{n}^k$ with $k=1,2$, the limiting critical values for  90\%, 95\% and 99\%  are 1.2, 1.657 and 2.8  respectively.

Since $\tilde{S}_{n}^1$  and $\tilde{S}_{n}^2$ admit the same limiting distribution, there is no need for the normalization by $L_1$ and $L_2$ in the definition of the combined statistics. Therefore, the combined statistics using the transformed process are defined as follows: $\tilde{S}_{n}^\star=\tilde{S}_{n}^1+\tilde{S}_{n}^2$, $\tilde{S}_{n}^\bullet=-2[\ln(P_v((\tilde{S}_{n}^1))+\ln(P_v(\tilde{S}_{n}^2))]$  and $\tilde{S}_{n}^\circ=\max(\tilde{S}_{n}^1,\tilde{S}_{n}^2)$.

 The statistics $\tilde S_{n}^\star$, $\tilde S_{n}^\bullet$ and $\tilde S_{n}^\circ$ will only be distribution-free if $\tilde{S}^1$ and $\tilde{S}^2$ are independent. As argued earlier, this is the case if the covariance term $K_{12}=0$ or more precisely if $\mathbb{E}\left\{(X_i-m(\theta_0,I_{i-1}))^3|\mathcal{F}_{i-1}\right\}=0$ a.s.

As mentioned above, the statistics $\{S_{n}^k\}_{k=1,2}$  given in \eqref{SS1} converge weakly towards $S^k=\int_{0}^{1}[\widehat{\mathbb{D}}^k(F^{-1}(u))]^2 du$. Since the process $\widehat{\mathbb{D}}^k$ has a complicated covariance,  which depends on the unknown distribution $F(\cdot)$ and some unknown parameters, the p-values of the test statistic $S_n^k$ need to be approximated using either numerical techniques or re-sampling algorithms. A numerical approximation of the limiting distribution, using the approach introduced by \cite{DM}, is outlined in Subsection \ref{secnumapp} and a re-sampling algorithm based on multipliers bootstrap is described in Subsection \ref{secmult}.

\subsection{Numerical approximation of the asymptotic distribution of the test statistics }\label{secnumapp}
To approximate the limiting distribution of the statistics $\{S_{n}^k\}_{k=1,2}$ and $S_n^\star$ we follow the  numerical integration procedure given by \cite{DM}. We will only outline how the procedure applies to $S_n^\star$ since the application to $\{S_{n}^k\}_{k=1,2}$ is similar and much simpler. The main idea consists in using numerical integration to approximate $S^\star=\int_{0}^{1}\left[ \widehat{\mathbb{D}}^1(F^{-1}(u))^2/L_1+ \widehat{\mathbb{D}}^2(F^{-1}(u))^2/L_2\right] du$ by the quadratic form $Q_m=\sum_{k=1}^m a_k\left[\widehat{\mathbb{D}}^1(F^{-1}(u_k))^2/L_1+\widehat{\mathbb{D}}^2(F^{-1}(u_k))^2/L_2\right]$,
where $u_k$ are the quadrature nodes, $a_k$  are the quadrature coefficients
and $m$ is the number of quadrature points. Because $\widehat{\mathbb{D}}$ is a Gaussian process, $Q_m$ is a quadratic
form of $(2m)$ normal random variables. Imhof's characteristic function inversion procedure (\cite{IM}) is used
to compute the distribution function of $Q_m$. When applying the procedure to $S_n^1$ or $S_n^2$, $Q_m$ will be a quadratic form with $(m)$ normal random variables. Note that in the computation of the characteristic function the covariance operator $\widehat{\mathcal{K}}$ is replaced by its empirical estimate obtained by replacing $K$, $\Gamma_{\theta_0}$, $G$ and $\Sigma_0$ by their consistent estimators $\hat K$, $\hat{\Gamma}$, $\hat{G}$ and $\hat{\Sigma}_0$, where $\hat{K}(t)$ is a symmetric two by two matrix whose diagonal entries are $\hat{K}_k$ for $k=1,2,$ and whose off diagonal entry $\hat{K}_{12}$ is given by
$$\hat{K}_{12}(t)=\frac{1}{n}\; \sum_{i=1}^n W_{\theta_n}^1(X_i,I_{i-1})W_{\theta_n}^2(X_i,I_{i-1}) \1{\{X_{i-1}\leq t\}},$$ and where
$$\hat{G}(t)=\frac{1}{n}\ \sum_{i=1}^n W_{\theta_n}(X_i,I_{i-1})\phi^\star(X_i,I_{i-1},\theta_n) \1{\{X_{i-1}\leq t\}},$$
$$\hat{\Gamma}(t)=\frac{1}{n} \sum_{i=1}^n \dot{W}_{\theta_n}(X_i, I_{i-1})\1{\{X_{i-1}\leq t\}}$$ and
$\hat{\Sigma}_0=\frac{1}{n}\ \sum_{i=1}^n\phi^\star(X_i,I_{i-1},\theta_n)\phi^\star(X_i,I_{i-1},\theta_n)^\top .$ We also used, in the simulation and application sections, the simplest quadratures, namely $a_k=1/m$ and $u_k=k/m$.

For the statistics $S_n^\bullet$ is computed after obtaining the p-values of statistics $S_n^1$ and $S_n^2$. The p-value of  $S_n^\bullet$ are obtained form the Chi-square distribution which is only valid when $S^1$ and $S^2$ are independent. The independence argument also allows for easy computation of the p-values of the statistics $S_n^\circ$. If the independence argument is not valid, then the computation for both $S_n^\bullet$ and $S_n^\circ$ are more complex and requires the numerical approximation of the joint distribution of $S_n^1/L_1$ and $S_n^2/L_2$.

\subsection{Multipliers bootstrap for the test statistics}\label{secmult}
Another approach to approximate the limiting distributions of $S^1_n, S_n^2, S_n^\star, S_n^\bullet$ and $S_n^\circ$ consists in using a re-sampling algorithm. To define such algorithm, let $B$ be a positive integer denoting the number of bootstrap samples. For $b=1,\dots, B$,  let  $\left( Z_{1,b}, \dots, Z_{n,b} \right)$ be a sequence of independent identically distributed random variables with mean zero and variance one independent of the sigma field generated by the $X_i$'s.  The multipliers bootstrap  technique applied to the empirical process $\mathbb{D}_n$, is defined as follows
\begin{eqnarray}\label{Dstar}
\mathbb{D}_{n,b}^{\star}(x) := \frac{1}{\sqrt{n}}\; \sum_{i=1}^n Z_{i,b} \biggl\{ W_{\theta_0}(X_i,I_{i-1}) \1{\{X_{i-1}\leq x\}}
-  \Gamma^\top_{\theta_0}(x)\phi^\star(X_i,I_{i-1},\theta_0)  \biggr\}.
\end{eqnarray}
Observe that the process $\mathbb{D}_n^{\star k}(x)$ defined in \eqref{Dstar} depends on unknown quantities, namely, $\theta$ and $\Gamma_{\theta}$. A plug-in estimator of $\mathbb{D}_n^{\star k}(x)$ is obtained by replacing $\theta$ and  $\Gamma_\theta$  by their consistent estimators $\theta_n$ and ${\Gamma}_{\theta_n}(x)$ yielding
\begin{eqnarray}\label{Dtilde}
\check{\mathbb{D}}_{n,b}(x) := \frac{1}{\sqrt{n}}\; \sum_{i=1}^n Z_{i,b} \biggl\{ W_{\theta_n}(X_i,I_{i-1}) \1{\{X_{i-1}\leq x\}}
- \Gamma^\top_{\theta_n}(x) \phi^\star(X_i,I_{i-1},\theta_n) \biggr\}.
\end{eqnarray}
The following extra assumptions are needed to establish the weak convergence of $\check{\mathbb{D}}_n^{k}$.\\
\noindent{\bf Assumption M1:} Assume that there exists $\delta>0$ such that for all $\theta$ satisfying $\|\theta-\theta_0\|\le \delta$ one has $\|\dot{W}^k_\theta(X_i,I_{i-1})- \dot{W}^k_{\theta_0}(X_i,I_{i-1}) \|\le \|\theta-\theta_0\|M_3(X_i,I_{i-1})$ where $M_3$ is a positive function satisfying $\mathbb{E}(M_3(X_i,I_{i-1}))\le C<\infty$.\\
\noindent{\bf Assumption M2:} Assume that there exists $\delta>0$ such that for all $\theta$ satisfying $\|\theta-\theta_0\|\le \delta$ one has $\|\phi^\star(X_i,I_{i-1},\theta)-\phi^\star(X_i,I_{i-1},\theta_0)-(\theta-\theta_0)^\top \dot{\phi}^\star(X_i,I_{i-1},\theta_0)\|\le \|\theta-\theta_0\|\lambda_3(\|\theta-\theta_0\|)M_4(X_i,I_{i-1})$ where $M_4$ and $\lambda_3$ are positive functions satisfying $\mathbb{E}(M_4(X_i,I_{i-1}))\le C <\infty$ and $\lim_{t\to 0}\lambda_3(t)=0$.\\
\noindent{\bf Assumption M3:} Assume that $E\|\dot{\phi}^\star(X_i,I_{i-1},\theta_0)\dot{\phi}^\star(X_i,I_{i-1},\theta_0)^\top\|\le C<\infty$.\\
The next result summarizes the asymptotic behavior of the bootstrapped process $\check{\mathbb{D}}_{n,b}$.
\begin{theorem}\label{theoremMB} Suppose that assumptions A1--A4 and M1--M3 hold true, then for $b=1,\ldots, B$,
$\check{\mathbb{D}}_{n,b}$ converge weakly to independent copies of $\widehat{\mathbb{D}} = \mathbb{D} - \Gamma^\top_{\theta_0}\Theta$.
\end{theorem}
Note that the bootstrapped version of the statistic $S^k_n$ for $k=1,2$ is given by
$$S_{n,b}^k := \int_{-\infty}^{\infty}\check{\mathbb{D}}_{n,b}^{k}(x) ^2F_n(dx)=\frac{1}{n}\sum_{i=1}^n\sum_{j=1}^n Z_{i,b}Z_{j,b} \mathcal{M}^k_{i,j}$$ where
\begin{multline*}
\mathcal{M}^k_{i,j}=W^k_{\theta_n}(X_i,I_{i-1})W^k_{\theta_n}(X_j,I_{j-1})(1 - F_n(X_{i-1}\vee X_{j-1})) \\ -W^k_{\theta_n}(X_i,I_{i-1})\phi^\star(X_j,I_{j-1},\theta_n)^\top\mathcal{L}^k_n(X_{i-1})\\
-W^k_{\theta_n}(X_j,I_{j-1})\phi^\star(X_i,I_{i-1},\theta_n)^\top\mathcal{L}^k_n(X_{j-1})
+ \phi^\star(X_i,I_{i-1},\theta_n)^\top \mathcal{N} \phi^\star(X_i,I_{i-1},\theta_n),
\end{multline*}
where $\mathcal{L}_n^k(t):=\int_t^\infty \Gamma^k_{\theta_n}(u)F_n(du)$ and $\mathcal{N}:=\int \Gamma^k_{\theta_n}(u) \Gamma^k_{\theta_n}(u)^\top F_n(du)$.

Let $L_{n,k}= \sum_{i=1}^n\hat{K}_k(X_i)/n$ for $k=1,2$. To obtain a bootstrapped version of
$S_n^\star$ and $S_n^\circ$ one uses  $S_{n,b}^\star=S_{n,b}^1/L_{n,1}+S_{n,b}^2/L_{n,2}$ and $S_{n,b}^\circ=\max\{S_{n,b}^1/L_{n,1},S_{n,b}^2/L_{n,2}\}$ respectively.

To apply the multipliers bootstrap procedure to approximate the p-value of any of the Cram\'er-von Mises statistics $S^1_{n}$,$S^2_{n}$, $S^\star_{n}$ or $S^\circ_{n}$ one proceeds according to the following algorithm. The algorithm is explained in the context of $S^1_n$ but works in the same manner for any of the statistics mentioned above.
\begin{itemize}
	\item Estimate $\theta$ by $\theta_n$  and compute the test statistic $S^1_{n}$.
	\item For  $b\in \{ 1, \dots, B\}$
	\begin{itemize}
		\item Generate $(Z_{1,b},\ldots,Z_{n,b})$ as sequence of i.i.d normal $(0,1)$ random variables.
		\item Compute the bootstrapped Cram\'er-von Mises statistics $S_{n,b}^1$ for $b=1,\ldots, B$.
  \end{itemize}
\item Estimate the $p$-value, $Pv_1$ of $S^1_{n}$ by $Pv_1=\frac{1}{B}\sum_{b=1}^B \1{\{S_{n,b}^1 > S^1_{n} \}}$.
	\end{itemize}
The statistics $S_n^\bullet=-2\ln(Pv_1)-2\ln(Pv_2)$ where  $Pv_k$ for $k=1,2,$ are the estimated $p$-values given by the above procedure. The $p$-value of the statistics $S_n^\bullet$ are obtained using the Chi-square distribution with $4$ degrees of freedom.

\begin{remark}
Note that $\theta$ is estimated only once and the matrix $\mathcal{M}$ is only computed once which make the multipliers Bootstrap extremely fast to implement. This is much faster than the classical parametric bootstrap which requires re-estimating the parameters for each bootstrap iteration.  Moreover, Theorem \ref{theoremMB} and the continuous mapping Theorem imply that $\left(S_{n, b}^1, \dots, S_{n, B}^1 \right)$ converges in law to independent copies of $S^1$.
\end{remark}
\section{Finite sample performance}\label{sim}
This section presents several simulation experiments carried out to assess the power of the proposed test statistics.
The first experiment, described in Subsection \ref{exp1},  is designed to assess the fit of an ARCH(1) model. The second experiment, outlined in Subsection \ref{exp2}, assess the power of our test statistics for detecting departure from GARCH(1,1) model. The third experiment given in Subsection \ref{exp3} studies the behavior of the test when fitting and AR(1)-GARCH(1,1) model.
The fourth experiment, presented in Subsection \ref{exp4} outlines the finding of applying our test procedures to the case of stochastic differential equation models. Two sub-experiments are considered, in the first we test if the model has a linear drift and a constant diffusion while in the second we test if the model has linear drift and a diffusion proportional of square root of series.

\subsection{Testing ARCH(1) model}\label{exp1}
The purpose of the simulation experiment considered here is to test whether the data is generated from  pure ARCH(1) model or not. That is we wish to test
\begin{eqnarray}
 &(H_0):& \mu(\cdot)=0 \mbox{ and } \sigma^2(\cdot)=\alpha_0+\alpha_1 X_{t-1}^2 \
 \quad \mbox{ versus } \nonumber \\
&(H_1):& \quad \mu(\cdot)\neq 0 \mbox{ or }\sigma^2(\cdot)\neq \alpha_0+\alpha_1 X_{t-1}^2 .
 \end{eqnarray}
 The data is generated according to one of the following  models,
$X_t = \sqrt{h_t}\;\epsilon_t,
$
where $\epsilon_t$ are independent and identically distributed $N(0,1)$ random variables and
\begin{itemize}
\item[($M_0$)] $h_t=1.1 + 0.5 X_{t-1}^{2}$
\item[($M_1$)] $h_t = 1.1 +0.5 X_{t-1}^2 + 0.5 X_{t-1}$
\item[($M_2$)] $h_t = 1.1 +0.5 X_{t-1}^2 + 0.5\; \mbox{sign}(X_{t-1})$
\item[($M_3$)] $h_t = 1.1 +0.5 X_{t-1}^2 +  X_{t-1}$
\item[($M_4$)] $h_t = 1.1 +0.5 X_{t-1}^2 +  \mbox{sign}(X_{t-1}).$
\end{itemize}

A pure ARCH(1) model is fitted to the data. Then tests described in Section \ref{sec4} are applied using the $\phi^\star$ corresponding to maximum likelihood estimator of the ARCH parameters $\alpha_0$ and $\alpha_1$. We generated series of length $n=100, 300$ following models $M_0$--$M_4$.
The results given in Table \ref{tabarch1} summarize a Monte-Carlo simulation study with 2000 replications of tests with $5\%$ significance levels. Note that $M_0$ corresponds to the null hypothesis and as shown in the table, the test maintains its  level quite well.  Tests  based on Khamaladze transform are in general a bit less powerful than those based on the original process.
Table \ref{tabarch1} also shows that $S_n^1$ and $\tilde S_n^1$ have no detection power in this context. This is expected since there is no change in the mean function for all the alternative considered here. The power of the combined statistics $S_n^\star, S_n^\circ$ and $S_n^\bullet$ (or $\tilde S_n^\star, \tilde S_n^\circ$ and $\tilde S_n^\bullet$ ) are similar in general. The powers obtained using the multipliers procedure and those derived from numerical approximation are very close. This shows that both techniques provide very good estimation of p-values of these test statistics.
\begin{table}[hbt]
\caption{\label{tabarch1}{ Percentage of rejection of the null hypothesis (ARCH(1) model) when the data are generated according to models $M_0$ to $M_4$ and $n = 100$ and $300$.} }
\begin{center}
{\scriptsize
\begin{tabular}{cl|rrrrr|rrrrr|rrrrr}
\hline
& &\multicolumn{5}{c|}{Transformation technique}  &\multicolumn{5}{c|}{Multipliers procedure} & \multicolumn{5}{c}{Numerical approximation} \\
\cline{3-17}
n& DGP&  {$\tilde S_{n}^1$} &  {$\tilde S_{n}^2$} & $\tilde S_{n}^\star$ & $\tilde S_{n}^\circ$  & $\tilde S_{n}^\bullet$
 &{$ S_{n}^1$} &  {$ S_{n}^2$} & $ S_{n}^\star$ & $S_{n}^\circ$  & $ S_{n}^\bullet$
 &{$ S_{n}^1$} &  {$ S_{n}^2$} & $ S_{n}^\star$ & $S_{n}^\circ$  & $ S_{n}^\bullet$ \\
\hline
&$(M_0)$	&	3.7	&	4.2	&	3.9	&	4.0	&	4.0	&	4.3	&	3.6	&	3.7	&	3.8	&	4.0	&	4.7	&	3.3	&	3.6	&	3.8	&	3.8	\\
&$(M_1)$	&	4.5	&	20.8	&	15.0	&	14.5	&	15.1	&	5.3	&	21.2	&	12.0	&	10.8	&	14.0	&	5.5	&	18.4	&	11.6	&	10.2	&	13.3	\\
100&$(M_2)$ &	4.3	&	20.9	&	14.9	&	14.9	&	14.8	&	4.8	&	15.2	&	9.5	&	7.9	&	11.8	&	5.2	&	14.2	&	9.0	&	7.6	&	11.3	\\
&$(M_3)$	&	3.8	&	43.4	&	30.7	&	32.9	&	29.4	&	5.2	&	61.0	&	42.2	&	44.7	&	44.5	&	5.3	&	60.4	&	42.6	&	44.5	&	43.1	\\
&$(M_4)$	&	3.1	&	47.4	&	35.1	&	35.9	&	34.3	&	3.6	&	43.5	&	24.8	&	24.1	&	32.5	&	3.9	&	42.6	&	24.5	&	24.0	&	31.0	\\
\hline
&$(M_0)$	&	3.9	&	5.5	&	4.6	&	5.0	&	4.7	&	4.4	&	5.0	&	4.5	&	4.5	&	4.8	&	4.4	&	4.6	&	4.3	&	4.7	&	4.4	\\
&$(M_1)$	&	3.7	&	45.3	&	35.9	&	37.0	&	35.1	&	4.7	&	57.7	&	44.3	&	44.5	&	45.3	&	4.9	&	57.8	&	44.3	&	45.1	&	44.3	\\
300&$(M_2)$ &	3.8	&	40.2	&	31.7	&	31.6	&	30.5	&	4.6	&	38.0	&	23.8	&	22.2	&	29.3	&	4.1	&	37.9	&	23.6	&	21.2	&	27.6	\\
&$(M_3)$	&	3.5	&	63.1	&	54.6	&	56.4	&	53.9	&	4.3	&	92.1	&	88.9	&	89.8	&	87.0	&	4.2	&	92.2	&	89.3	&	90.0	&	86.9	\\
&$(M_4)$	&	3.6	&	68.0	&	59.2	&	60.4	&	58.4	&	4.3	&	76.0	&	64.6	&	64.7	&	66.9	&	4.4	&	75.9	&	64.7	&	64.9	&	66.7	\\
\hline
\end{tabular}
}
\end{center}
 \end{table}

\subsection{Testing GARCH(1,1) model}\label{exp2}
This section presents the result of a simulation study in which we assess the power of the tests, discussed in this manuscript, in detecting departure from a GARCH(1,1) model. That is we wish to test is the mean and variance functions are those of a GARCH(1,1) or not.
We use the same settings as in  Experiment 2 of \cite{escanciano2010}. More precisely, we generated the data from an AR(1)-GARCH(1,1) model with autoregressive coefficient $a_1$ varying from $-0.9$ to $0.9$. We fitted a GARCH(1,1) and recorded the percentage of rejection. As in all simulation reported in this manuscript we used $2000$ Monte-Carlo simulation iterations. For this experiment we used a sample size $n=100$ an in \cite{escanciano2010}. The results are reported in Table \ref{tabgarch}. Note that $a_1=0$ corresponds to the null hypothesis in this case. Table \ref{tabgarch} shows that the $5\%$ level is respected quite well in general. It also shows that the type of alternative considered here is mainly detected by the contribution of component $S_n^1$ or $\tilde S_n^1$ to the combined statistics. This makes all combined statistics quite power in detecting such alternatives.  The component $S_n^2$ and $\tilde S_n^2$ have no power against this type of alternatives.
This is expected since these components were designed to detect change in the variance function.
In this context statistics based on the transformed process seem to be a bit more powerful than those based on the original process. Comparing our results to Figure 1 of \cite{escanciano2010}, we notice that the power of the tests presented here is slightly better.

\begin{table}[hbt]
\caption{\label{tabgarch}{ Percentage of rejection of the null hypothesis (GARCH(1,1) model) when the data are generated according to AR(1)-GARCH(1,1) with autoregressive coefficient $a_1$.} }
\begin{center}
{\scriptsize
\begin{tabular}{cc|rrrrr|rrrrr|rrrrr}
\hline
& &\multicolumn{5}{c|}{Transformation technique}  &\multicolumn{5}{c|}{Multipliers procedure} & \multicolumn{5}{c}{Numerical approximation} \\
\cline{3-17}
n& {$a_1$ }&  {$\tilde S_{n}^1$} &  {$\tilde S_{n}^2$} & $\tilde S_{n}^\star$ & $\tilde S_{n}^\circ$  & $\tilde S_{n}^\bullet$
 &{$ S_{n}^1$} &  {$ S_{n}^2$} & $ S_{n}^\star$ & $S_{n}^\circ$  & $ S_{n}^\bullet$
 &{$ S_{n}^1$} &  {$ S_{n}^2$} & $ S_{n}^\star$ & $S_{n}^\circ$  & $ S_{n}^\bullet$ \\
\hline
&	-0.9	&	100	&	1.3	&	100	&	100	&	100	&	99.4	&	2.5	&	96.2	&	96.2	&	99.1	&	99.4	&	1.9	&	95.9	&	96.0	&	99.2	\\
&	-0.7	&	100	&	3.0	&	100	&	100	&	100	&	98.7	&	1.5	&	96.4	&	96.9	&	97.7	&	98.7	&	1.2	&	96.2	&	96.7	&	97.8	\\
&	-0.5	&	98.9	&	2.9	&	97.8	&	97.9	&	97.7	&	93.4	&	2.4	&	85.3	&	87.5	&	86.9	&	93.3	&	2.1	&	84.4	&	87.0	&	86.5	\\
&	-0.3	&	76.3	&	2.9	&	63.8	&	66.0	&	62.3	&	59.6	&	3.1	&	43.9	&	45.6	&	46.3	&	59.7	&	2.4	&	42.5	&	45.1	&	45.2	\\
100&	0.0	&	6.3	&	2.9	&	4.4	&	4.2	&	4.1	&	3.6	&	4.1	&	3.4	&	3.4	&	4.1	&	3.5	&	4.0	&	2.9	&	3.2	&	3.8	\\
&	0.3	&	68.1	&	2.1	&	54.2	&	57.9	&	52.4	&	42.4	&	2.2	&	26.8	&	31.2	&	26.7	&	42.4	&	1.5	&	25.4	&	30.1	&	24.7	\\
&	0.5	&	99.0	&	2.5	&	96.8	&	97.3	&	96.3	&	80.7	&	1.4	&	67.7	&	72.4	&	65.6	&	81.4	&	1.5	&	67.4	&	72.0	&	65.3	\\
&	0.7	&	100	&	1.5	&	100	&	100	&	100	&	87.9	&	0.5	&	79.8	&	83.9	&	77.8	&	87.9	&	0.4	&	79.8	&	83.6	&	77.5	\\
&	0.9	&	100	&	1.0	&	100	&	100	&	100	&	80.6	&	0.3	&	62.1	&	69.6	&	69.8	&	80.9	&	0.2	&	62.0	&	68.5	&	69.9	\\
\hline
&	-0.9	&	100	&	2.9	&	100	&	100	&	100	&	100	&	4.1	&	99.8	&	99.8	&	99.9	&	100	&	3.9	&	99.8	&	99.8	&	99.9	\\	
&	-0.7	&	100	&	4.7	&	100	&	100	&	100	&	100	&	3.8	&	99.9	&	100	&	100	&	100	&	3.1	&	100	&	100	&	100	\\	
&	-0.5	&	100	&	4.9	&	100	&	100	&	100	&	99.7	&	3.1	&	99.3	&	99.5	&	99.5	&	99.7	&	3.0	&	99.3	&	99.5	&	99.5	\\	
&	-0.3	&	99.3	&	4.6	&	98.4	&	98.6	&	98.2	&	95.6	&	3.1	&	91.5	&	92.9	&	90.9	&	95.8	&	2.7	&	91.8	&	92.9	&	91.0	\\	
300&	0.0	&	4.8	&	4.1	&	4.7	&	4.4	&	4.5	&	4.4	&	4.4	&	4.2	&	4.6	&	4.5	&	4.4	&	4.3	&	4.0	&	4.2	&	4.2	\\	
&	0.3	&	99.2	&	3.1	&	97.8	&	98.5	&	97.5	&	94.3	&	2.5	&	88.8	&	92.1	&	88.3	&	94.5	&	2.6	&	89.3	&	91.5	&	88.4	\\	
&	0.5	&	100	&	2.9	&	100	&	100	&	100	&	99.6	&	2.2	&	99.3	&	99.5	&	99.4	&	99.6	&	2.0	&	99.3	&	99.5	&	99.4	\\	
&	0.7	&	100	&	2.0	&	100	&	100	&	100	&	99.7	&	1.3	&	99.5	&	99.6	&	99.5	&	99.7	&	0.8	&	99.4	&	99.5	&	99.5	\\	
&	0.9	&	100	&	0.5	&	100	&	100	&	100	&	96.6	&	0.3	&	90.7	&	94.1	&	93.1	&	96.6	&	0.2	&	90.6	&	93.6	&	92.9	\\	
\hline
\end{tabular}
}
\end{center}
 \end{table}

\subsection{Testing AR(1)-GARCH(1,1) model}\label{exp3}
In this section we present the result of a study testing the null hypothesis that the mean and variance functions are those of an  AR(1)-GARCH(1,1) model. The same five alternatives considered in Experiment 3 of \cite{escanciano2010} are used. To be specific the data is generated according to one of these alternatives, and AR(1)-GARCH(1,1) model is fitted to the data. The test statistics proposed in this manuscript are applied and the percentage of rejection, in a $2000$ replicates Monte-Carlo simulation,  is reported in Table \ref{tabexp3}. The alternative considered are defined as follows:
\begin{description}
\item[A0:] The null hypothesis AR(1)-GARCH(1,1) model:  $X_t=0.02 + 0.02 X_{t-1}+\varepsilon_t$ where $\varepsilon_t=\sqrt{h_t}u_t$ with $h_t=0.08+0.1\varepsilon_{t-1}^2+0.85 h_{t-1}$ and $u_t$ i.i.d Normal random variables with mean zero and variance one.
\item[A1:] ARMA(1,1)-GARCH(1,1) model: $X_t=0.02 + 0.02 X_{t-1}+0.5\varepsilon_{t-1} +\varepsilon_t$ and $\varepsilon_t$ is as in model A0.
\item[A2:] TAR model: $X_t=0.6X_{t-1}+\varepsilon_t$ if $X_{t-1}\le 1$ and $X_t=-0.5X_{t-1}+\varepsilon_t$ if $X_{t-1}>1$ with $\varepsilon_t$ is as in model A0.
\item[A3:] EGARCH(1,1) model: $X_t=\sqrt{h_t}u_t$ where $\ln(h_t)=0.025 +0.5 \ln(h_{t-1})+0.25(|u_{t-1}|-\sqrt{2/\pi})-0.8u_{t-1}$  and  $u_t$ is as in model A0.
\item[A4:] Bilinear model: $X_t=0.6X_{t-1} +0.7 u_{t-1}X_{t-2}+u_t$ where   $u_t$ is as in model A0.
\item[A5:] Nonlinear Moving average model: $X_t=0.8u_{t-1}^2 +u_t$ where   $u_t$ is as in model A0.
\end{description}
\begin{table}[hbt]
\caption{\label{tabexp3}{ Percentage of rejection of the null hypothesis (AR(1)-GARCH(1,1) model) when the data are generated according to A0--A5.} }
\begin{center}
{\scriptsize
\begin{tabular}{cc|rrrrr|rrrrr|rrrrr}
\hline
& &\multicolumn{5}{c|}{Transformation technique}  &\multicolumn{5}{c|}{Multipliers procedure} & \multicolumn{5}{c}{Numerical approximation} \\
\cline{3-17}
n& DGP &  {$\tilde S_{n}^1$} &  {$\tilde S_{n}^2$} & $\tilde S_{n}^\star$ & $\tilde S_{n}^\circ$  & $\tilde S_{n}^\bullet$
 &{$ S_{n}^1$} &  {$ S_{n}^2$} & $ S_{n}^\star$ & $S_{n}^\circ$  & $ S_{n}^\bullet$
 &{$ S_{n}^1$} &  {$ S_{n}^2$} & $ S_{n}^\star$ & $S_{n}^\circ$  & $ S_{n}^\bullet$ \\
\hline
&A0&4.6	&	3.9	&	4.6	&	4.5	&	4.9	&	2.6	&	5.4	&	4.4	&	5.3	&	4.5	&	2.3	&	5.3	&	4.5	&	5.1	&	4.2	\\
&A1&3.6	&	4.1	&	3.4	&	3.8	&	3.6	&	1.4	&	4.3	&	3.2	&	3.5	&	2.9	&	1.3	&	3.8	&	2.9	&	3.4	&	2.8	\\
&A2&100	&	5.8	&	99.8	&	99.9	&	99.8	&	77.8	&	31.9	&	61.6	&	47.2	&	78.8	&	77.7	&	31.9	&	63.1	&	47.3	&	78.4	\\
300&A3&6.2	&	6.0	&	7.1	&	6.8	&	7.1	&	55.5	&	51.6	&	71.8	&	64.9	&	78.1	&	55.5	&	52.0	&	71.7	&	64.9	&	76.9	\\
&A4&68.3	&	5.3	&	59.0	&	60.3	&	57.7	&	13.6	&	27.6	&	30.5	&	27.2	&	32.2	&	13.0	&	27.4	&	31.1	&	27.3	&	31.7	\\
&A5&100	&	15.4	&	100	&	100	&	100	&	91.3	&	18.1	&	87.2	&	87.4	&	89.0	&	91.4	&	18.2	&	87.7	&	86.7	&	89.4	\\
\hline
&A0&	6.2	&	3.4	&	4.8	&	4.7	&	4.7	&	3.5	&	4.4	&	3.8	&	4.1	&	3.9	&	3.2	&	4.7	&	3.9	&	4.5	&	3.8	\\
&A1&	4.0	&	4.9	&	3.7	&	3.8	&	3.6	&	2.2	&	5.0	&	4.2	&	4.5	&	3.3	&	2.2	&	4.7	&	3.8	&	4.4	&	3.1	\\
&A2&	100	&	11.2	&	100	&	100	&	100	&	98.2	&	63.7	&	96.0	&	91.0	&	98.4	&	98.1	&	63.6	&	96.1		&91.0	&	98.4	\\
600&A3&	4.8	&	13.7	&	10.8	&	11.2	&	10.6	&	80.1	&	69.6	&	92.6	&	90.5	&	94.6	&	80.7	&	70.0	&	92.6	&	90.4	&	94.8	\\
&A4&	88.4	&	6.4	&	84.0	&	84.4	&	83.3	&	30.7	&	41.7	&	51.6	&	44.0	&	56.6	&	29.5	&	41.1	&	50.8	&	43.3	&	55.9	\\
&A5&	100	&	40.9	&	100	&	100	&	100	&	97.6	&	50.2	&	97.5	&	97.2	&	97.9	&	97.6	&	51.0	&	97.5	&	97.3	&	97.8	\\
\hline
\end{tabular}
}
\end{center}
 \end{table}
Comparing with Table 2 in \cite{escanciano2010}, one notice that combined tests introduced in this manuscript are more powerful than the tests considered in \cite{escanciano2010}. The only exception being the case of alternative A1: ARMA(1,1)-GARCH(1,1) where our tests were not able to detect such alternative while those in \cite{escanciano2010} had reasonable power.
We also notice that  tests based on the transformed process are bit more powerful in the case of alternatives A2, A4 and A5. On the other hand these tests fall way behind in the case A3. This seems to concord with the observations from the first simulation experiment where tests based on  the transformed process were a bit more powerful in detecting a change in the mean function. In fact, for the three alternatives A2, A4 and A5 there is a change in the conditional mean function. While in the case of alternative A3, the conditional mean function is still that of an AR(1) while the change occurred in the conditional variance function.

\subsection{Testing  stochastic differential equation models}\label{exp4}
In this part, we are interested in testing whether the data fits a specified  stochastic differential equation (SDE) model. In finance continuous time models are widely used to study the dynamic of some financial products such as asset prices, interest rate and bonds. Continuous-time modelling in finance was introduced by  \cite{Bach} on the theory of speculation and really started with \cite{merton1970} seminal work. Since then several models were developed which can be formulated by the following general SDE
\begin{eqnarray}\label{SDE}
dX_t = \mu_\theta(X_t) dt + \sigma(X_t) dW_t,
\end{eqnarray}
where $W_t$ is a standard Brownian motion. The   drift $\mu_\theta(\cdot)$ and the diffusion $\sigma^2(\cdot)$ are known functions. Several well-known models in financial econometrics, including \cite{BS}, \cite{Vas}, \cite{CIR}, \cite{CKLS} and \cite{ait}, among others, can be written under the form \eqref{SDE} with a specific form of drift and diffusion functions. Below  the list of SDE models considered here.
\begin{itemize}
\item[N1:] \texttt{Vasicek}: $dX_t = 3(10-X_t) dt + 5dW_t, \quad x_0 = 0.03$
\item[N2:] \texttt{Hyperbolic}: $dX_t = 5\frac{X_t}{\sqrt{1+X_t^2}} dt + 5dW_t\quad x_0=3$
\item[N3:] \texttt{CIR}: $dX_t = (1+4.5 X_t)dt + 0.75\sqrt{X_t}dW_t, \quad x_0=3$
\item[N4:] \texttt{CKLS1}: $dX_t = 1.5(1 - X_t) dt + 1.5 X_t^{0.8}dW_t, \quad x_0 =5$
\item[N5:] \texttt{CKLS2}: $dX_t = 1.5(1 - X_t) dt + 0.5 X_t^{1.5}dW_t, \quad x_0 =5$
\item[N6:] \texttt{A\"it-Sahalia}: $dX_t = (1 + 15 X_t +0.25 X_t^{-1} - 2X_t^2) dt +0.5 X_t^{1.5}dW_t, \quad x_0=5$
\end{itemize}
In practice, the diffusion process $\{X_t\}$ is observed at instants $\{t =i \Delta  | i=0, \dots, n\}$, where $\Delta >0$ is generally small, but fixed as $n$ increases. For instance the series could be observed hourly, daily, weekly or monthly. Therefore, we may model these discretely observed measurements using time series models. In fact, despite the fact that \eqref{SDE} is written in a continuous-time form, one often uses the following Euler discretization scheme to get
\begin{eqnarray}\label{SDE2}
X_{t+\Delta} - X_t = \mu_\theta(X_t) \Delta + \sigma(X_t)\left(W_{t+\Delta} - W_t \right), \quad t=0, \Delta, 2\Delta, \dots,
\end{eqnarray}
 as an approximation that facilitates computational and theoretical derivation. The accuracy of such Euler discretisation is studied in \cite{JP}.  The purpose of this simulation study is to generate processes satisfying the SDEs described by models N1--N6 above, then use the procedures described in Section \ref{sec4} to test whether the data fits a specific type of SDE models.
 Two types of null hypotheses will be considered. In the first sub-experiment we test if the data fits a \texttt{Vasicek} type SDE model, namely we test if $\mu$ is of the form $a+b X_{t-1}$ for some parameters $(a,b)\in \mathbb{R}^2$ and that $\sigma=c$ for some parameter $c>0$. Note that only model N1 belongs to this class of models. In the second sub-experiment the null hypothesis considered is a \texttt{CIR} model where $\mu$ is of the form  $a+bX_{t-1}$ and $\sigma$ of the form $c\sqrt{X_{t-1}}$. For this second sub-experiment, N3 belongs to the null hypothesis.

 We suppose the process is observed over the time interval $[0, T]$ and $n$ corresponds to the number of instants where the process is observed.  The sampling mesh in such case is $\Delta=T/n$. In order to assess the sensitivity of our test to the sampling frequency from the underlying process, we consider $T=1$ and $\Delta= 0.002, 0.005$ and $0.01$, corresponding to a total sampling instants $n=500, 200$ and $100$ respectively.  Such framework supposes that collecting more observations means shortening the time interval between successive existing observations, not lengthening the time period over which data are recorded.\\
 The data is generated according to one of  SDE models listed above. The  parameters  are estimated, after discretization, using the maximum pseudo-likelihood method (see \cite{ait2}). Table \ref{tabSDE1} and \ref{tabSDE2}  report the percentage of rejection of the two null hypotheses considered. The results are  obtained using $1000$ Monte-Carlo replications for different values of sampling mesh $\Delta$.

For the first sub-experiment, Table \ref{tabSDE1} shows that the tests are quite powerful in detecting change in the diffusion term. Alternatives N3--N6 are rejected with high probability even for $n=100$. The three combined statistics are more or less similar in term of their detection power. As expected the statistics $S_n^1$ or $\tilde S_n^1$ do no detect change in the diffusion term. All statistics fail to detect the alternative N2 even for large sample sizes.

For the second sub-experiment, reported in Table \ref{tabSDE2}, tests based on the original process and using either the multipliers procedure or the numerical approximation do respect their levels more appropriately than those based on the transformed process.  The  three combined statistics have the power to detect all alternatives considered here with N5 and N6 being easier to detect. This  could be explained by the large change in the diffusion functions between N5 or N6 and the null hypothesis in this case N3.
\begin{table}[hbt]
 \caption{\label{tabSDE1}{Percentage of rejection of the null hypothesis $(H_0)$ (Vasicek type Model) for different sample sizes $n=100, 200$ and $500$.}}
\begin{center}
{\scriptsize
\begin{tabular}{cc|rrrrr|rrrrr|rrrrr}
\hline
& &\multicolumn{5}{c|}{Transformation technique}  &\multicolumn{5}{c|}{Multipliers procedure} & \multicolumn{5}{c}{Numerical approximation} \\
\cline{3-17}
n& DGP &  {$\tilde S_{n}^1$} &  {$\tilde S_{n}^2$} & $\tilde S_{n}^\star$ & $\tilde S_{n}^\circ$  & $\tilde S_{n}^\bullet$
 &{$ S_{n}^1$} &  {$ S_{n}^2$} & $ S_{n}^\star$ & $S_{n}^\circ$  & $ S_{n}^\bullet$
 &{$ S_{n}^1$} &  {$ S_{n}^2$} & $ S_{n}^\star$ & $S_{n}^\circ$  & $ S_{n}^\bullet$ \\
\hline
&	N1	&	6.1	&	4.1	&	4.7	&	4.6	&	5.1	&	6.1	&	4.1	&	3.8	&	4.1	&	4.7	&	5.2	&	3.7	&	3.5	&	3.6	&	4	\\
&	N2	&	5.1	&	5.9	&	5.2	&	4.9	&	5.3	&	6.6	&	4.2	&	4.2	&	4.4	&	5.9	&	5.8	&	4.1	&	4.3	&	4.1	&	4.5	\\
100&	N3	&	2.4	&	95	&	84.8	&	87.8	&	82.5	&	4.9	&	98.4	&	98.3	&	98.4	&	94.1	&	4.9	&	98.6	&	98.3	&	98.6	&	93.3	\\
&	N4	&	6.1	&	84.4	&	73.1	&	76.1	&	70.9	&	4	&	90.7	&	90	&	90.7	&	79.4	&	3.7	&	90.8	&	90.2	&	90.8	&	77.3	\\
&	N5	&	9.2 &92.4& 82.8& 82.9& 81.9& 4.9& 91.5& 90.6& 91.5& 80.7& 5.1& 92.3& 91.5& 92.2& 79.9\\
& N6  & 11.3 &88.9& 80.3& 79.3& 78.6& 10.2& 91.1& 89.4& 91.1& 77.9& 9.3& 91.3& 90& 91.3& 75.9 \\
\hline
&N1	&	5.2	&	5.5	&	5.4	&	5.2	&	5.2	&	6.9	&	4.7	&	4.5	&	5.2	&	5.4	&	7.6	&	4.6	&	4	&	4.8	&	4.7	\\
&N2	&	5.9	&	6.1	&	4.9	&	5.6	&	4.8	&	6.5	&	3.9	&	3.8	&	3.9	&	5.4	&	6.3	&	3.9	&	4.1	&	4	&	4.7	\\
200	&N3	&	4.1	&	100	&	99.8	&	99.9	&	99.5	&	5.6	&	100	&	100	&	100	&	99.6	&	5.7	&	100	&	100	&	100	&	99.7	\\
	&N4	&	7.5	&	96.5	&	94.7	&	94.4	&	94.1	&	4.1	&	98.2	&	98.2	&	98.2	&	96.1	&	3.9	&	98.2	&	98.1	&	98.2	&	95.8	\\
&	N5	& 10.7 &99.5& 97.8& 98.7& 97.4& 5.1& 98.9& 98.7& 98.9& 95.9& 4.8& 98.7& 98.6& 98.7& 96.8\\
&   N6  & 11.4 &99.4& 98.3& 98.7& 97.8& 6& 99.6& 99.6& 99.6& 97& 6.3& 99.6& 99.5& 99.6& 97.3\\
\hline		
	&N1	&	4.9	&	5.5	&	5.3	&	5.4	&	6.4	&	8.2	&	4.8	&	5.6	&	5.1	&	7.7	&	8.3	&	4.3	&	5.5	&	4.7	&	6.3	\\
	&N2	&	5.1	&	6.4	&	5.6	&	5.7	&	5.8	&	7.1	&	4.3	&	4.6	&	4.4	&	5.3	&	6.8	&	4.4	&	4.6	&	4.5	&	4.7	\\
500	&N3	&	4.1	&	100	&	100	&	100	&	100	&	5.5	&	100	&	100	&	100	&	100	&	5.2	&	100	&	100	&	100	&	100	\\
	&N4	&	8.6	&	100	&	100	&	100	&	100	&	6.4	&	100	&	100	&	100	&	100	&	6	&	100	&	100	&	100	&	100	\\
	&N5	&	11.8 &100 &100 &100 &100 &  4.4 &100 & 99.8 & 99.9 & 99.9 &  4.4 &100 & 99.8 & 99.8 & 99.9 \\
	&N6	&   13.7 &100 &100 &100 &100 &  7.6 &100 &100 &100 & 99.9 &  7.8 &100 &100 &100 & 99.9 \\
\hline
  \end{tabular}
  }
  \end{center}
  \end{table}

\begin{table}[hbt]
 \caption{\label{tabSDE2}{Percentage of rejection of the null hypothesis $(H_0)$ (CIR type Model) for different sample sizes $n=100, 200$ and $500$.}}
\begin{center}
{\scriptsize
\begin{tabular}{cc|rrrrr|rrrrr|rrrrr}
\hline
& &\multicolumn{5}{c|}{Transformation technique}  &\multicolumn{5}{c|}{Multipliers procedure} & \multicolumn{5}{c}{Numerical approximation} \\
\cline{3-17}
n& DGP &  {$\tilde S_{n}^1$} &  {$\tilde S_{n}^2$} & $\tilde S_{n}^\star$ & $\tilde S_{n}^\circ$  & $\tilde S_{n}^\bullet$
 &{$ S_{n}^1$} &  {$ S_{n}^2$} & $ S_{n}^\star$ & $S_{n}^\circ$  & $ S_{n}^\bullet$
 &{$ S_{n}^1$} &  {$ S_{n}^2$} & $ S_{n}^\star$ & $S_{n}^\circ$  & $ S_{n}^\bullet$ \\
\hline
&	N1	&	4.8	&	50.1	&	38.9	&	39.6	&	38.9	&	4.4	&	30.9	&	27.1	&	30.5	&	23.2	&	3.6	&	31.4	&	26	&	29.7	&	20.7	\\
&	N2	&	4.9	&	63.3	&	53.9	&	53.5	&	52.9	&	4.3	&	42.2	&	37.5	&	41.3	&	30.4	&	4	&	42.3	&	36	&	40.4	&	28.9	\\
100&N3	&    4.1	&	10.2	&	7.3	&	8.2	&	6.4	&	4.5	&	3.9	&	3.8	&	4.2	&	4.9	&	4	&	4.9	&	3.8	&	4.5	&	4.4	\\
&	N4	&	5.4	&	10.5	&	9.4	&	7	&	10.1	&	3.8	&	24.6	&	19.2	&	23.6	&	16.5	&	3.3	&	25.5	&	19.2	&	23.8	&	16.3	\\
&	N5	&  10.6& 54& 40.7& 35.7& 41.7& 5.8& 70.1& 65.9& 69.7& 52.2& 5.3& 71.9& 66.1& 70.7& 51.7 \\
& N6 & 10.2& 51.1& 42.9& 37.5& 44& 7.6& 53.8& 50.5& 53.7& 39.5& 6.4& 55.2& 50& 55& 37.6 \\
\hline 	
&   N1	&5.4	&	71.2	&	62.2	&	63.6	&	61	&	5.5	&	55.7	&	50.8	&	52.7	&	45.1	&	4.7	&	54.9	&	49.6	&	52.1	&	44.4	\\
&   N2	&4.9	&	85.5	&	79.9	&	80.6	&	79.1	&	5.8	&	77.6	&	70.9	&	75.8	&	64.3	&	5.2	&	76.1	&	70	&	74.7	&	63.9	\\
200	&N3 &5.9	&	9.8	&	8.5	&	9.6	&	8.2	&	5.3	&	4.7	&	4.5	&	4.5	&	5.3	&	5.7	&	4.6	&	4.4	&	4.1	&	4.5	\\
	&N4	&	6.3	&	24.1	&	19.3	&	17.5	&	19.7	&	3.9	&	51.4	&	46.6	&	49.6	&	39.7	&	4.2	&	50.4	&	45.5	&	49.1	&	38.5	\\
&	N5  & 11.7 & 88.2 & 79.2 & 77.5 & 79.0 &  5.1 & 95.4 & 95.1 & 95.3 & 89.0 &  4.6 & 95.4 & 95.1 & 95.4 & 88.3 \\
&  N6&  12.2 & 84.3 & 75.7 & 73.3 & 74.8 &  7.1 & 87.7 & 87.0 & 87.7 & 76.9 &  6.8 & 88.1 & 86.8 & 88.1 & 77.6 \\
\hline		
	&N1	&	7.6	&	94.3	&	91.6	&	91.8	&	90.9	&	6.7	&	88.4	&	85.2	&	87	&	84.2	&	6.9	&	89.1	&	85.6	&	86.9	&	83.5	\\
	&N2	&	7.1	&	98.4	&	97.1	&	96.8	&	97.3	&	4.6	&	96.3	&	94.7	&	96	&	93	&	4.4	&	96.5	&	94.5	&	96	&	92.2	\\
500	&N3 &4.8	&	9.3	&	7.8	&	8	&	7.6	&	4.6	&	4.3	&	3.9	&	4.3	&	3.9	&	4.5	&	3.8	&	3.4	&	3.8	&	4	\\
	&N4	&	7.5	&	64.8	&	56.3	&	54.5	&	55.9	&	4.3	&	82.1	&	81.3	&	81.9	&	75.9	&	4.3	&	82	&	81.5	&	81.6	&	75.7	\\
	&N5	& 13.0 & 99.9 & 99.4 & 99.3 & 99.3 &  6.0 &100.0 & 99.7 & 99.8 & 99.8 &  5.7 & 99.7 & 99.7 & 99.7 & 99.6 \\
& N6  & 14.9 & 99.7 & 99.7 & 99.4 & 99.6 &  9.2 & 99.8 & 99.7 & 99.8 & 99.5 &  8.7 & 99.8 & 99.7 & 99.8 & 99.5 \\
\hline
  \end{tabular}
  }
  \end{center}
  \end{table}

\section{Application to interest rate data}\label{app}
This section presents as application the testing methodology to interest rate dynamics. Monthly values of interest rate for a maturity of 1 months (\% per year) from July 1964 to April 1989 are considered. The monthly interest rate $(r_t)$ dynamic is displayed in Figure \ref{TS}. Figure \ref{change} is a scatter plot of monthly changes in interest rate against the previous month's rate and it clearly shows an obvious heteroscedasticity, with the range of changes increasing significantly as the level of interest rates increases. An augmented Dickey-Fuller test of non stationarity of interest rate data is performed and the non stationarity hypothesis is rejected, with p-value equal to 0.01, for 1 to 12 months lagged series.
\begin{figure}[hbt]
\centering
\includegraphics[width=8cm, height=6cm]{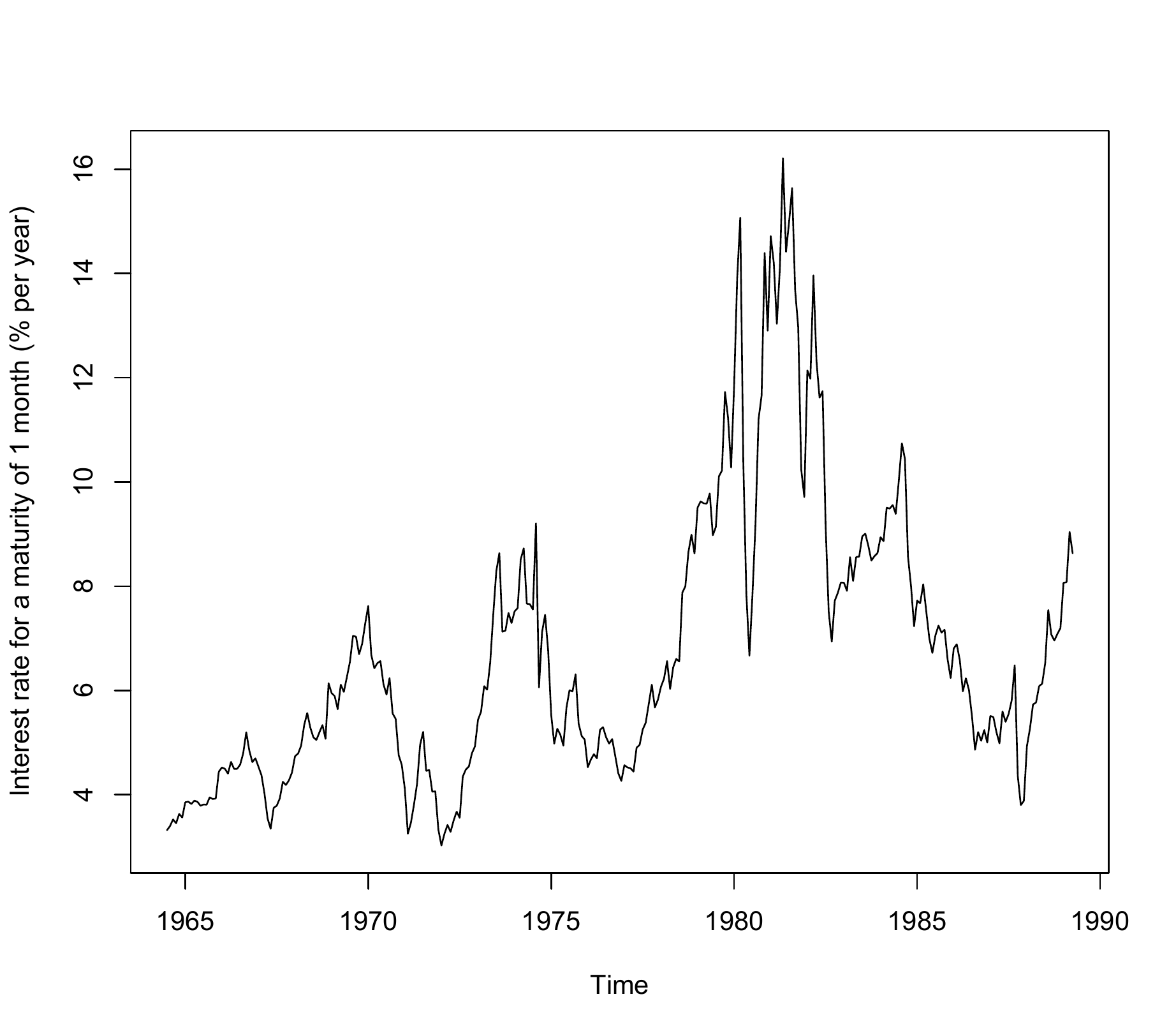}
\caption{Interest rate for a maturity of 1 months (\% per year) from July 1964 to April 1989. }
\label{TS}
\end{figure}
\begin{figure}[hbt]
\centering
\includegraphics[width=10cm]{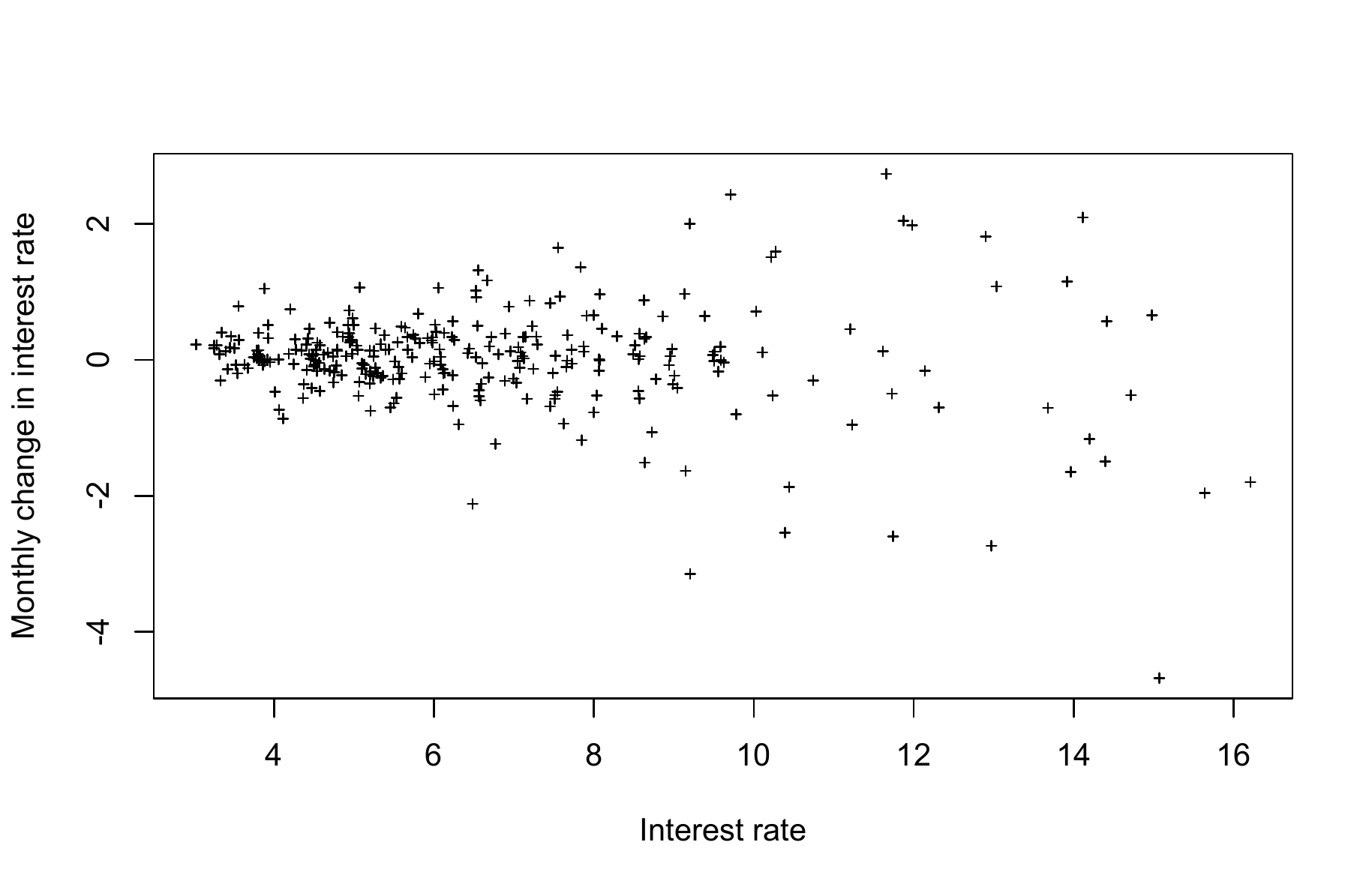}
\caption{Monthly changes in interest rate against interest rate on preceding month. }
\label{change}
\end{figure}

The aim here is to verify if the interest rate data described above fits a selected diffusion model. The following   models will be considered as candidates
\begin{itemize}
\item[\texttt{(D1)}] \texttt{Vasicek}: $dr_t = (\alpha + \beta r_t) dt + \sigma dW_t $
\item[\texttt{(D2)}] \texttt{Hyperbolic}: $dr_t = \alpha\frac{r_t}{\sqrt{1+r_t^2}} dt + \sigma dW_t$
\item[\texttt{(D3)}] \texttt{A\"it-Sahalia 1}: $dr_t = (\alpha_0 + \alpha_1 r_t + \alpha_2 r_t^{-1} + \alpha_3 r_t^2) dt + \sigma dW_t$
\item[\texttt{(D4)}] \texttt{CIR}: $dr_t = (\alpha+\beta r_t)dt + \sigma\sqrt{r_t}dW_t$
\item[\texttt{(D5)}] \texttt{CKLS 1}: $dr_t = \kappa(\alpha - r_t) dt + \sigma r_t^{0.8}dW_t$
\item[\texttt{(D6)}] \texttt{CKLS 2}: $dr_t = \kappa(\alpha - r_t) dt + \sigma r_t^{1.5}dW_t$
\item[\texttt{(D7)}] \texttt{A\"it-Sahalia 2}: $dr_t = (\alpha_0 + \alpha_1 r_t + \alpha_2 r_t^{-1} + \alpha_3 r_t^2) dt + \sigma r_t^{1.5}dW_t$
\end{itemize}

The procedure consists in fitting  the interest rate data to each of the  models outlined above and then  apply the statistics described  of Section \ref{sec4} to see how good is the model fit.
Table \ref{tabIR} reports the p-values resulting from applying the statistics  discussed in Section \ref{sec4} after fitting each of the models D1--D7 to the monthly interest rate data.  As expected,  all models with  constant conditional volatility (D1-D3) are clearly rejected even when we considered different drift functions.
\begin{table}[hbt]
 \caption{\label{tabIR}{P-values of the test statistics applied to different models (D1--D7) fitted to the monthly interest rate data.}}
\begin{center}
{\scriptsize
\begin{tabular}{l|c|ccccccc}
\hline
&Test&	D1	&	D2	&	D3	&	D4	&	D5	&	D6	&	D7	\\ \cline{2-9}
&$\tilde S_{n}^1$&	0.1432	&	0.1108	&	0.3589	&	0.1548	&	0.1574	&	0.1779	&	0.5179	\\
&$\tilde S_{n}^2$&	0.0001	&	0.0001	&	0.0000	&	0.0018	&	0.0096	&	0.7873	&	0.7937	\\
Transformation technique&$\tilde S_{n}^\star$&	0.0001	&	0.0002	&	0.0002	&	0.0028	&	0.0128	&	0.3911	&	0.7779	\\
&$\tilde S_{n}^\circ$ &	0.0001	&	0.0002	&	0.0001	&	0.0036	&	0.0191	&	0.3241	&	0.7676	\\
&$\tilde S_{n}^\bullet$&	0.0001	&	0.0002	&	0.0002	&	0.0026	&	0.0113	&	0.4153	&	0.7765	\\ \hline
&$ S_{n}^1$&	0.2840	&	0.3080	&	0.5130	&	0.5370	&	0.6410	&	0.8170	&	0.9850	\\
&$ S_{n}^2$&	0.0000	&	0.0000	&	0.0000	&	0.0000	&	0.0020	&	0.6080	&	0.5920	\\
Multipliers procedure&$ S_{n}^\star$&	0.0000	&	0.0000	&	0.0000	&	0.0000	&	0.0020	&	0.7800	&	0.8350	\\
&$S_{n}^\circ$&	0.0000	&	0.0000	&	0.0000	&	0.0000	&	0.0020	&	0.7060	&	0.6680	\\
&$ S_{n}^\bullet$&	0.0000	&	0.0000	&	0.0000	&	0.0000	&	0.0098	&	0.8443	&	0.8976	\\ \hline
&$ S_{n}^1$&	0.2566	&	0.2893	&	0.5424	&	0.5302	&	0.6522	&	0.8286	&	0.9868	\\
&$ S_{n}^2$&	0.0000	&	0.0001	&	0.0000	&	0.0001	&	0.0010	&	0.6077	&	0.6126	\\
Numerical approximation&$ S_{n}^\star$&	0.0000	&	0.0002	&	0.0000	&	0.0002	&	0.0019	&	0.7812	&	0.8343	\\
&$S_{n}^\circ$&	0.0000	&	0.0001	&	0.0000	&	0.0001	&	0.0010	&	0.7169	&	0.6908	\\
&$ S_{n}^\bullet$ &	0.0001	&	0.0003	&	0.0001	&	0.0008	&	0.0056	&	0.8490	&	0.9088	\\ \hline
  \end{tabular}
  }
  \end{center}
  \end{table}

Models D4 and D5 are also rejected, implying that the rate in the diffusion function of the monthly interest rate is not  $0.5$ or $0.8$.
On the other hand, all  tests  cannot reject the D6 and D7 implying  that the diffusion of the interest rate is proportional to $X_t^{1.5}$.
This conclusion is in perfect concordance with existing findings in the literature (see for instance \cite{ait3}). In fact, the rate of  $1.5$  was the recommended choice in \cite{ait}, among others, in a study of the same interest rate data.

\section{Conclusion}\label{conc}
Though Khmaldze matringale transform is quite popular as it simplifies the asymptotic behavior, its finite sample behavior is in general far from being the best compared to re-sampling or numerical approximation of the original statistics. We also notice that even the level of tests based on Khmaladze transform was sensitive to the estimation of the function $g$, defined in \eqref{Kh1}.  In our case, the estimation function $g$ used a non-parametric  estimate of the derivative of the function $\Gamma_\theta$ and we noticed that the level was sensitive to the bandwidth parameter utilized in the estimation. The level of the tests obtained by numerical approximation or re-sampling of the original statistics are stable and close to their target values. Tests, using the original cumulative residual process, are in general more powerful than the ones based on martingale transformation. Moreover, as mentioned earlier, p-values of the combined statistics based on the transformed process are only computable when the components $\mathbb{D}^1$ and $\mathbb{D}^2$ are independent. Whereas p-values for the multipliers procedure for both $S_n^\star$ and $S_n^\circ$ are obtained in the same manner whether the components  $\mathbb{D}^1$ and $\mathbb{D}^2$ are independent or not. For the approximation technique, the algorithm described in the manuscript provides the p-values of $S_n^\star$  whether the components  $\mathbb{D}^1$ and $\mathbb{D}^2$ are independent or not. For $S_n^\circ$  we used the independence of $\mathbb{D}^1$ and $\mathbb{D}^2$ to simplify our computation, but the process can be generalized to the dependent case by approximating the distribution of the maximum of two dependent quadratic forms of normal random variables.
The computations involved in the approximation techniques  are extremely fast. As mentioned in the paper, the computation  involved in the multipliers bootstrap  are  quite fast since the parameters are estimated only once and that each bootstrap iteration just requires generating $n$ i.i.d normal mean zero and variance one random variables. The combined statistics $S_n^\star$, $S_n^\circ$ and $S_n^\bullet$ have similar power behavior in general, but $S_n^\star$ is much easier to compute and its p-values are easier to obtain.  We, therefore, recommend practitioners to use tests based on the original process more often with a multipliers procedure or a numerical approximation technique. Among these statistics $S_n^\star$ would be the easiest to implement. The tests introduced in this manuscript are in general a bit more powerful than those in \cite{escanciano2010}, but there are alternatives, like A1, where tests in \cite{escanciano2010} clearly outperform those discussed here.

The procedures discussed here  generalize to the case of multivariate time series where $X_i\in \mathbb{R}^d$. The procedure based on the multipliers bootstrap would be the easiest to generalize. In fact, it will not require any modification it suffices to adjust the definition of the process and the functions $\phi^\star$ and $\Gamma_\theta$ to the multivariate case. The generalization of the numerical approximation and that of the martingale transform would need more technical work.

\appendix
\section{Proofs}
We start by proving the next Lemma which is used repeatedly in the proofs. It establishes a  uniform law of large number result needed to show that the convergence is uniform in $x$ for both theorems (\ref{Tn}) and (\ref{theoremMB}).
\begin{lemma}\label{lem_ULLN}
If $\chi_i$, $i=1\ldots,n$ is strictly stationary and ergodic series satisfying $\mathbb{E}|\chi_1|<\infty$ and if $X_i$ is strictly stationary and ergodic series then, we have
\begin{enumerate}
\item [i)] $\sup_{x\in \mathbb{R}}\left|\frac{1}{n}\sum_{i=1}^n\chi_i\1\{X_{i-1}\le x\}-\mathbb{E}[\chi_1\1\{X_{0}\le x\}]\right|
$ converges to zero almost surely,
\item [ii)] If $g_k$ and $A_k$ satisfy the conditions of Theorem \ref{Tn} then 
\begin{multline*}\sup_{x\le x_0}\biggl\|\frac{1}{n}\sum_{i=1}^n\chi_i \int_{-\infty}^{\min(x, X_{i-1})} (g_k(y))^\top A_k^{-1} (y)
K_{k}(dy) \\ -\mathbb{E}\biggl[ \chi_1\int_{-\infty}^{\min(x, X_{0})}  (g_k (y))^\top A_k^{-1} (y)
K_{k}(dy)\biggr]\biggr\|
\end{multline*} converges to zero almost surely for $k=1,2$ and for every $x_0$ satisfying $A_k(x_0)$ invertible.
\end{enumerate}
\end{lemma}
\begin{proof}
The proof of i)   and ii)  are quite similar. Here only the proof of ii) is given. For i) one repeats the same steps. One also can see that i) is similar to 4.1 in \cite{KS}.  Though ii) can be deduced from of ULLN of \cite{And} if the set of $x$'s is compact. Since this is not the case here, a direct proof using a Glivenko-Cantelli type argument shall be given next. First recall that $A_k(x)$ is invertible for all $x\le x_0$ whenever $A_k(x_0)$ is invertible and that one can easily verify that $\|A_k^{-1}(x)\|\le \|A_k^{-1}(x_0)\|<\infty$. One also sees that
\begin{eqnarray*}
\left\| \int_{-\infty}^{\min(x, X_{0})}  (g_k(y))^\top A_k^{-1} (y)
K_{k}(dy)\right\| \le \left[ \int_{-\infty}^{x_0} \|(g_k(y))^\top \|
K_{k}(dy)\right]\| A^{-1} (x_0)\|<\infty,\end{eqnarray*}
since $g_k$ satisfies the condition of Theorem \ref{Tn}.
 The LLN for stationary ergodic sequence yields the almost sure convergence of ii) for every fixed $x\le x_0$. To prove that the convergence is uniform in $x\le x_0$ one uses a Glivenko-Cantelli type argument applied using $\eta^\star(x)=\mathbb{E}\left[ |\chi_1|\int_{-\infty}^{\min(x, X_{0})}  \|(g_k(y))^\top A_k^{-1} (y)\|K_{k}(dy)\right]$.  Note that $\eta^\star$ is a continuous increasing function and for all $x\le x_0$ one has $\eta^\star(x)\le \eta^\star(x_0)<\infty$. Therefore for every $\epsilon >0$ there exist a finite partition $-\infty=t_0<t_1<\ldots t_k=x_0$ such that $0\le \eta^\star(t_{j+1})-\eta^\star(t_j)\le \epsilon$. To ease presentation, let
$$\eta_n(x)=\frac{1}{n}\sum_{i=1}^n\chi_i\int_{-\infty}^{\min(x, X_{i-1})} (g_k(y))^\top A_k^{-1} (y)K_{k}(dy)$$ and $$\eta(x)=\mathbb{E}\left[ \chi_1\int_{-\infty}^{\min(x, X_{0})}  (g_k(y))^\top A_k^{-1} (y)K_{k}(dy)\right].$$
 Note that for any $x\le x_0$ there exists $j< k$ such that $t_j< x\le t_{j+1}$. Clearly, $\|\eta_n(x)-\eta(x)\|\le \|\eta_n(x)-\eta_n(t_j)\|+\|\eta_n(t_j)-\eta(t_j)\|+\|\eta(t_j)-\eta(x)\|.$  Using its definition one sees that
 \begin{eqnarray*}\|\eta_n(x)-\eta_n(t_j)\|&\le& \biggl| \frac{1}{n}\sum_{i=1}^n|\chi_i|\int_{\min(t_{j}, X_{i-1})}^{\min(t_{j+1}, X_{i-1})} \|(g_k(y))^\top A_k^{-1} (y)\|K_{k}(dy)\\
 & & \quad -\eta^\star(t_{j+1})+\eta^\star(t_{j})\biggr| 
 +|\eta^\star(t_{j+1})-\eta^\star(t_{j})|.
 \end{eqnarray*}
  One also easily verifies that
$\|\eta(t_j)-\eta(x)\|\le |\eta^\star(t_{j+1})-\eta^\star(t_{j})|$. Therefore
\begin{multline*}
\sup_{x\le x_0}\left|\eta_n(x)-\eta(x)\right| \le \|\eta_n(t_j)-\eta(t_j)\|+2\epsilon +\max_{j=0,\ldots,k} \biggl| \frac{1}{n}\sum_{i=1}^n|\chi_i|\\
\times\int_{\min(t_{j}, X_{i-1})}^{\min(t_{j+1}, X_{i-1})} \|(g_k(y))^\top A_k^{-1} (y)\| K_{k}(dy)-\eta^\star(t_{j+1})+\eta^\star(t_{j})\biggr|.
\end{multline*}
 The pointwise LLN and the fact that $k$ is finite implies that the first and last terms in the above inequality converge almost surely to zero. Since $\epsilon$ was arbitrary, one concludes that $\sup_{x\le x_0}\left|\eta_n(x)-\eta(x)\right|$ converges to zero almost surely.
\end{proof}
\subsection{Proof of Theorem \ref{TheoremA}}
The weak convergence of $\mathbb{D}_n$ follows from Theorem 1 in \cite{escanciano2007b}.
For $\widehat{\mathbb{D}}_n$, direct manipulations show that
\begin{eqnarray*}
\widehat{\mathbb{D}}_n(x)= \mathbb{D}_n(x)-\Gamma^\top_{\theta_0}(x)\sqrt{n}(\theta_n-\theta_0)+B_{n,1}(x)+B_{n,2}(x)
\end{eqnarray*}
where $$B_{n,1}(x)=\left[\frac{1}{n}\sum_{i=1}^n\dot{W}^\top_{\theta_0}(X_i,I_{i-1})\1\{X_{i-1}\le x\}- \Gamma^\top_{\theta_0}(x)\right]\sqrt{n}(\theta_n-\theta_0),$$ and
 $$B_{n,2}(x)=\frac{1}{\sqrt{n}}\sum_{i=1}^n\left[W_{\theta_n}(X_i,I_{i-1})-W_{\theta_0}(X_i,I_{i-1})- \dot{W}^\top_{\theta_0}(X_i,I_{i-1})(\theta_n-\theta_0)\right]\1\{X_{i-1}\le x\}.$$
Note that
$\|B_{n,1}\|$ converges in probability to zero by Lemma \ref{lem_ULLN}, the definition of  $\Gamma_{\theta_0}$ and the fact that $\sqrt{n}(\theta_n-\theta_0)=O_\mathbb{P}(1)$. By  Assumption (A3), the term   $\|B_{n,2}\|$ is bounded by $\sqrt{n}\|\theta_n-\theta_0\|\lambda_1(\|\theta_n-\theta_0\|)\frac{1}{n}\sum_{i=1}^nM_1(X_i,I_{i-1})$ which converges to zero in probability by the LLN and the fact that $\sqrt{n}\|\theta_n-\theta_0\|=O_\mathbb{P}(1)$.   Therefore, $\widehat{\mathbb{D}}_n(x)$ is asymptotically equivalent to ${\mathbb{D}}_n(x)-\Gamma^\top_{\theta_0}(x)\sqrt{n}(\theta_n-\theta_0)$.  Calling on  Assumption (A2),  one verifies that $\sqrt{n}(\theta_n-\theta_0)$ is tight and converges to $\Theta$ and that $(\mathbb{D}_n,\sqrt{n}(\theta_n-\theta_0))$ converge jointly to $(\mathbb{D},\Theta)$. Hence $\widehat{\mathbb{D}}_n$ converges to $\widehat{\mathbb{D}}={\mathbb{D}}-\Gamma^\top_{\theta_0}(x)\Theta$. Straightforward computations show the covariance function of $\widehat{\mathbb{D}}$ is precisely given by \eqref{def9}.\qed

\subsection{Proof of Theorem \ref{localt}}
The proof follows the same approach as that of Theorem \ref{TheoremA}. Precisely, one writes
\begin{eqnarray*}
	\widehat{\mathbb{D}}_n(x)= \mathbb{D}_n(x)+B_{n,0}(x)-B_{n,1}(x)+B_{n,2}(x)
\end{eqnarray*}
where
$$B_{n,0}(x)=\frac{1}{n}\sum_{i=1}^n\left(\begin{array}{l}a_1(I_{i-1})\\
a_2(I_{i-1})\end{array}\right)\1\{X_{i-1}\le x\},$$
$$B_{n,1}(x)=\left[\frac{1}{n}\sum_{i=1}^n\dot{W}^\top_{\theta_0}(X_i,I_{i-1})\1\{X_{i-1}\le x\}\right]\sqrt{n}(\theta_n-\theta_0),$$ and
$$B_{n,2}(x)=\frac{1}{\sqrt{n}}\sum_{i=1}^n\left[W_{\theta_n}(X_i,I_{i-1})-W_{\theta_0}(X_i,I_{i-1})- \dot{W}^\top_{\theta_0}(X_i,I_{i-1})  (\theta_n-\theta_0)\right]\1\{X_{i-1}\le x\}.$$
$B_{n,0}$ converges uniformly to $\Psi_A(x)$ by Lemma \ref{lem_ULLN}. Assumption (L1) and Lemma \ref{lem_ULLN}, yield that  $B_{n,1}$ converges to $\Gamma^\top_{\theta_0}(x) (\Theta +\xi_A)$.
From Assumption (A3), one concludes that the term $B_{n,2}$ is uniformly  bounded by $\sqrt{n}\|\theta_n-\theta_0\|\lambda_1(\|\theta_n-\theta_0\|)\frac{1}{n}\sum_{i=1}^nM_1(X_i,I_{i-1})$ which goes to zero in probability since
$\mathbb{E}\{M_1(X_i,I_{i-1})\}<\infty$ and by Assumption (L3), $\sqrt{n}\|\theta_n-\theta_0\|=O_{\mathbb{P}}(1)$ and $\|\theta_n-\theta_0\|=o_{\mathbb{P}}(1)$.\qed

\subsection{Proof of Theorem \ref{Tn}}
The proof of Theorem (\ref{Tn}) is as follows. First, in Lemma \ref{Tn1} we establish, for $k=1,2,$ that $T_n^k(\widehat{\mathbb{D}}_n^k)$ is asymptotically equivalent to $T^k(\widehat{\mathbb{D}}_n)$. Second,  using the continuous mapping theorem, one concludes that  $T^k(\widehat{\mathbb{D}}_n^k)$ converges to $T^k(\widehat{\mathbb{D}}^k)$. Then the proof is completed by showing that $T^k(\widehat{\mathbb{D}}^k)$ is equal to $T^k(\mathbb{D}^k)$  which has the same law as $\mathbb{D}^k$.
\begin{lemma}\label{Tn1}
Under the assumptions of Theorem \ref{Tn},
$\sup_{x\le x_0}\|T_n^k(\widehat{\mathbb{D}}_n^k)(x)-T^k(\widehat{\mathbb{D}}_n^k)(x)\|$ converges to zero in probability for $k=1,2$.
\end{lemma}
\begin{proof}  To prove the Lemma,  observe that
\begin{multline*}
T_n^k(\widehat{\mathbb{D}}_n^k)(x)-T^k(\widehat{\mathbb{D}}_n^k)(x)
= \int_{-\infty}^x g_k^\top(y)A_k^{-1}(y) \left( \int_y^\infty g_k(z){\widehat{\mathbb{D}}_n^k(dz)}\right)K_k(dy)\\
 -\int_{-\infty}^x \hat{g}_k^\top(y) \hat{A}_{k}^{-1}(y) \left( \int_y^\infty \hat{g}_k(z){\widehat{\mathbb{D}}_n^k(dz)}     \right)\hat{K}_{k}(dy).
\end{multline*}
First, one establishes that $\psi_n^k$ is tight and that $\sup_y|\psi_n^k(y)-\widehat\psi_n^k(y)|$ converges to zero in probability, where
$\psi_n^k(y)=\int_y^\infty g_k(z){\widehat{\mathbb{D}}_n^k(dz)}$ and $\widehat\psi_n^k(y)=\int_y^\infty \hat{g}_k(z){\widehat{\mathbb{D}}_n^k(dz)}$. For the tightness of $\psi_n^k$, set $\psi_n^{0,k}(y)= \int_y^\infty g_k(z){\mathbb{D}}_n^k(dz)=\frac{1}{\sqrt{n}}\sum_{i=1}^n W^k_{\theta_0}(X_i,I_{i-1})g_k(X_{i-1})\1\{X_{i-1}>y\}$. Note that $\psi_n^{0,k}$ is a marked empirical process its tightness follows, from \cite{escanciano2007b} and Assumption (K4).  Next, using the same decomposition as in the proof of Theorem \ref{TheoremA}, one sees that $$\sup_{x\in \mathbb{R}}\left|\psi_n^k(x)-\psi_n^{0,k}(x)+\frac{1}{n}\sum_{i=1}^n g_k(X_{i-1})\dot{W}_{\theta_0}^k(X_i,I_{i-1})^\top\1\{X_{i-1}>x\}\sqrt{n}(\theta_n-\theta)\right|=o_\mathbb{P}(1).$$ The above with tightness of  $\psi_n^{0,k}$ and assumptions (A3-A4, K4) imply  the tightness of $\psi_n^k$. For $\sup_y|\psi_n^k(y)-\widehat\psi_n^k(y)|$, the decomposition in the proof of Theorem \ref{TheoremA}, shows that
\begin{eqnarray*}
\lefteqn{\sup_y|\psi_n^k(y)-\widehat\psi_n^k(y)| \le  \left\|\int_y^\infty (\hat{g}_k(y)-g_k(y))\mathbb{D}_n^k(dy)\right\|} \\
& & + \left\|\frac{1}{n}\sum_{i=1}^n(\hat{g}_k(X_{i-1})-{g}_k(X_{i-1}))\dot{W}_{\theta_0}^k(X_i,I_{i-1})^\top\1\{X_{i-1}>x\}\sqrt{n}(\theta_n-\theta)\right\|
+o_\mathbb{P}(1),\\
&\le& \left\|\int_y^\infty (\hat{g}_k(y)-g_k(y))\mathbb{D}_n^k(dy)\right\| +\|\sqrt{n}(\theta_n-\theta)\| \|\hat{g}_k-g_k\|\frac{1}{n}\sum_{i=1}^n\|\dot{W}_{\theta_0}^k(X_i,I_{i-1})\|+o_\mathbb{P}(1)\\
\end{eqnarray*}
which converges to zero by assumptions and the LLN.

Next, it will be shown that $\sup_{x} \|\hat{A}_k(x)-A_k(x)\|$ converges to zero in probability. Observe that
\begin{multline*}
\|\hat{A}_k(x)-A_k(x)\| = \left\|\int_x^\infty g_k(t)g_k^\top(t) K_k(dt)-\int_x^\infty \hat{g}_k(t)\hat{g}_k^\top (t) \hat{K}_k(dt)\right\|\\
\le \left\|\int_{x}^\infty g_k(t)g_k^\top(t) (K_k(dt)-\hat{K}_k(dt))\right\|+\left\|\int_x^\infty \left[g_k(t)g_k^\top(t)- \hat{g}_k(t)\hat{g}_k^\top (t)\right] \hat{K}_k(dt)\right\|.
\end{multline*}
The first term above is equal $$\|\frac{1}{n}\sum_{i=1}^n g_k(X_{i-1})g_k^\top(X_{i-1})W^k_{\theta_n}(X_i,I_{i-1})\1\{X_{i-1}>x\} -\int_x^\infty  g_k(t)g_k^\top(t)K_k(dt)\|$$ which converges to almost surely zero by Lemma (\ref{lem_ULLN}) and Assumption (K2).
The second term is bounded by $(\|\hat{g}_k-g_k\|+\|g_k\|) \|\hat{g}_k-g_k\|\int_{-\infty}^\infty  \hat{K}_k(dt)$ which goes to zero in probability by assumptions. \\
As pointed earlier, for all $x\le x_0$ one has $A_k(x) -A_k(x_0)$  non-negative definite implying that $A_k(x)$ is invertible  whenever $A_k(x_0)$ is invertible and that $\|A_k^{-1}(x)\|\le \|A_k^{-1}(x_0)\|<\infty$. Using the above result for $\hat{A}_k-A_k$ and classical algebraic manipulations one also easily see that $\sup_{x\le x_0} \|\hat{A}_k^{-1}(x)-A_k^{-1}(x)\|\stackrel{\mathbb{P}}{\longrightarrow} 0$.\\
To complete the proof of the  Lemma \ref{Tn1}  note that
\begin{eqnarray*}
T_n^k(\widehat{\mathbb{D}}_n^k)(x)-T^k(\widehat{\mathbb{D}}_n^k)(x)= \int_{-\infty}^x g_k^\top(y)A_k^{-1}(y)\psi_n^k(y)K_k(dy)
-\int_{-\infty}^x \hat{g}_k^\top(y)\hat{A}_k^{-1}(y)\widehat{\psi}_n^k(y)\hat{K}_k(dy).
\end{eqnarray*}
Direct computations show that\\
\begin{multline*}
T_n^k(\widehat{\mathbb{D}}_n^k)(x)-T^k(\widehat{\mathbb{D}}_n^k)(x)=\int_{-\infty}^x  g_k^\top(y) A_k^{-1}(y)\psi_n(y)(K_k(dy)-\hat{K}_k(dy)) \\
+\int_{-\infty}^x \left[g_k^\top(y)A_k^{-1}(y)\psi_n^k(y)-\hat{g}_k^\top(y)\hat{A}_k^{-1}(y)\hat{\psi}_n^k(y)\right]\hat{K}_k(dy).
\end{multline*}
Hence
\begin{multline*}
\sup_{x\le x_0}|T_n^k(\widehat{\mathbb{D}}_n^k)(x)-T^k(\widehat{\mathbb{D}}_n^k)(x)| \le \sup_{x\le x_0}\left\|\int_{-\infty}^x  g_k^\top(y) A_k^{-1}(y)\psi_n(y)(K_k(dy)-\hat{K}_k(dy))\right\| \\
+\sup_{x\le x_0}\left\|\int_{-\infty}^x \left[g_k^\top(y)A_k^{-1}(y)\psi_n^k(y)-\hat{g}_k^\top(y)\hat{A}_k^{-1}(y)\hat{\psi}_n^k(y)\right]\hat{K}_k(dy)\right\|.
\end{multline*}
Using Lemma (4.1) in \cite{KS} and the fact that $\psi_n$ is tight, one concludes that the first term in the above goes to zero in probability.
Term 2  is bounded by $\sum_{\ell=1}^3\check{C}_{\ell,n}$ where
$$\check{C}_{1,n}=\|\widehat{\psi}_n^k-\psi_n^k\|\sup_{x\le x_0} \|A_k^{-1}(x)\|  \int_{-\infty}^\infty \|{g}_k^\top(y)\|\hat{K}_k(dy),$$
$$\check{C}_{2,n}=\sup_{x\le x_0}\|\hat{A}_k^{-1}(x)-A_k^{-1}(x)\|(\|\hat\psi_n^k-\psi_n^k\|+ \|\psi_n^k\|)\int_{-\infty}^\infty \|\hat{g}_k^\top(y)\|\hat{K}_k(dy)$$ and
$$\check{C}_{3,n}=\sup_{x\le x_0}(\|\hat{A}_k^{-1}(x)-A_k^{-1}(x)\|+\|A_k^{-1}(x)\|) (\|\hat\psi_n^k-\psi_n^k\|+ \|\psi_n^k\|) \|\hat{g}_k^\top-g_k^\top\|\int_{-\infty}^\infty \hat{K}_k(dy).$$ Since $\psi_n^k$ is tight and $\|\psi_n^k\|=O_{\mathbb{P}}(1)$, the proof is complete by noting that, using the assumptions and the above results, $\check{C}_{\ell,n}$ for $\ell=1,\ldots, 3$ converge to zero in probability.

The transformation $T^k$ is linear and continuous which implies that $T^k(\widehat{\mathbb{D}}_n^k)$ converges to $T^k(\widehat{\mathbb{D}}^k)=T^k(\mathbb{D}^k-\Theta^\top \Gamma_{\theta_0}^k )=T^k(\mathbb{D}^k)-\Theta^\top T^k(\Gamma_{\theta_0}^k)= T^k(\mathbb{D}^k)$ since direct computations show that $T^k(\Gamma_{\theta_0}^k)=0$. Straightforward computations, similar to those in \cite{stute-etal1998}, enable us to verify that $T^k(\mathbb{D}^k)$ has the same distribution as $\mathbb{D}^k=\mathbb{W}_k\circ K_k$ where  $\mathbb{W}_k$ is standard Brownian motion.
\end{proof}


\subsection{Proof of Theorem \ref{theoremMB}}
Theorem \ref{theoremMB} will be shown by establishing parts $(a)$ and $(b)$ below
\begin{itemize}
\item[$(a)$] $\check{\mathbb{D}}_{n,b}$ is asymptotically equivalent to $ \mathbb{D}_{n,b}^{\star}$
\item[$(b)$] $\mathbb{D}_{n,b}^{\star}$ converge in Law to independent copies of $\widehat{\mathbb{D}}.$
\end{itemize}
\noindent{\bf Proof of $(a)$:}\\
Straightforward computations show that
$\check{\mathbb{D}}^k_n(x) - \mathbb{D}_n^{\star k}(x) =\sum_{k=1}^5 R_{n,k}(x)$, where\\
$R_{n,1}(x)= \left(\frac{1}{n}\sum_{i=1}^n  Z_{i,b}\dot{W}^\top_{\theta_0}(X_i,I_{i-1})\1_{\{X_{i-1}\leq x \}}\right)\sqrt{n}(\theta_n-\theta_0) ,$\\
$R_{n,2}(x)=\frac{1}{\sqrt{n}}\sum_{i=1}^n Z_{i,b} \left[{W}_{\theta_n}(X_i,I_{i-1})-{W}_{\theta_0}(X_i,I_{i-1})- \dot{W}^\top_{\theta_0}(X_i,I_{i-1})(\theta_n-\theta_0)\right]\1_{\{X_{i-1}\leq x \}}, $\\
$R_{n,3}(x)=\left[\Gamma_{\theta_n}(x) - \Gamma_{\theta_0}(x) \right]^\top\left(\frac{1}{\sqrt{n}}\sum_{i=1}^n Z_{i,b} \phi^\star(X_i,X_{i-1},\theta_0)\right),$\\
$R_{n,4}(x)=\Gamma^\top_{\theta_n}(x)\frac{1}{\sqrt{n}}\sum_{i=1}^n Z_{i,b} \dot{\phi}^\star(X_i,X_{i-1},\theta_0)^\top(\theta_n-\theta_0),$ and\\
$R_{n,5}(x)=\Gamma^\top_{\theta_n}(x)\frac{1}{\sqrt{n}}\sum_{i=1}^n Z_{i,b} \biggl[\phi^\star(X_i,X_{i-1},\theta_n) -\phi^\star(X_i,X_{i-1},\theta_0)
- \dot{\phi}^\star(X_i,X_{i-1},\theta_0)^\top(\theta_n-\theta_0)\biggr].$\\
The term $R_{n,1}$ converges to zero in probability by Lemma \ref{lem_ULLN}, term $R_{n,2}$ converges in probability to zero uniformly for all $x\in \mathbb{R}$ following the same steps as for the term $B_{n,2}$ in the proof of Theorem \ref{TheoremA}. By Assumption (M1), $\sup_{x}\|\Gamma_{\theta_n}(x) - \Gamma_{\theta_0}(x)\|\le \|\theta_n - \theta_0\|C$ which goes to zero in probability. One also sees that $\sup_x\|\Gamma_{\theta_0}(x)\|\le \mathbb{E}\|\dot{W}_{\theta_0}(X_i,I_{i-1})\|<\infty$. Next, since $\frac{1}{\sqrt{n}}\sum_{i=1}^n Z_{i,b} \phi^\star(X_i,X_{i-1},\theta_0)=O_\mathbb{P}(1)$ by the multiplier central limit theorem, then the term $R_{n,3}$ converge uniformly  to zero in probability. $$\|R_{n,4}\|\le\|\theta_n-\theta_0\| \|\frac{1}{\sqrt{n}}\sum_{i=1}^n Z_{i,b} \dot{\phi}^\star(X_i,X_{i-1},\theta_0)^\top\| (\|\Gamma_{\theta_n} - \Gamma_{\theta_0}\|+\|\Gamma_{\theta_0}\|)=o_\mathbb{P}(1)$$ since $\frac{1}{\sqrt{n}}\sum_{i=1}^n Z_{i,b} \dot{\phi}^\star(X_i,X_{i-1},\theta_0)=O_\mathbb{P}(1)$ by Assumption (M3). Finally, by Assumption (M2),
$$\|R_{n,5}\|\le \sqrt{n}\|\theta_n-\theta_0\|\lambda_3(\|\theta_n-\theta_0\|)\left(\frac{1}{n}\sum_{i=1}^n|Z_{i,b}|M_4(X_{i},X_{i-1})\right)
\times(\|\Gamma_{\theta_n} - \Gamma_{\theta_0}\|+\|\Gamma_{\theta_0}\|)=o_\mathbb{P}(1).$$
 Combining these results, yields $\sup_x|\check{\mathbb{D}}^{n,b}(x) - \mathbb{D}_{n,b}{\star}(x)|=o_\mathbb{P}(1)$ which completes the proof of part $(a)$.\\
\noindent{\bf Proof of $(b)$:}\\
One notices that the multiplier central limit theorem yields the weak convergence of $\mathbb{D}_{n,b}^{\star}$. It just remain to show that the asymptotic covariance operator of the process $\mathbb{D}_{n,b}^{\star}$ is the same as that of the limiting process $\widehat{\mathbb{D}}.$
$$\mbox{Cov}\left( \mathbb{D}_{n,b}^{\star}(x), \mathbb{D}_{n,b}^{\star}(y) \right) = \frac{1}{n} \sum_{i=1}^n \sum_{j=1}^n \mathbb{E}(V_i(x) V_j(y))\mathbb{E}(Z_{i,b}Z_{j,b}),$$ where $V_i(x)={W}_{\theta_0}(X_i,I_{i-1}) \1{\{X_{i-1}\leq x\}}
- \phi^\star(X_i,I_{i-1},\theta_0)^\top \Gamma_{\theta_0}(x).$
One has only to consider the terms with $i=j$ in the above sums, since for $i\neq j, \E\left\{ Z_{j,b} Z_{i,b} \right\} = 0$. Recalling that  $\mathbb{E}([Z_{i,b}]^2)=1$, one  sees that the above covariance reduces to
$
\mbox{Cov}\left( \mathbb{D}_n^{\star k}(x), \mathbb{D}_n^{\star k}(y) \right) = \frac{1}{n} \sum_{i=1}^n \E\left( V_i(x) V_i(y) \right)=\E\left( V_1(x) V_1(y) \right).
$
Straightforward computations show that
$$\E\left( V_1(x) V_1(y)\right) = {K}(x,y)-\Gamma^\top_{\theta_0}(x) G(y)-G^\top(x)\Gamma_{\theta_0}(y)
+ \Gamma^\top_{\theta_0}(x)\Sigma_0\Gamma_{\theta_0}(y),$$
which is the same as the covariance of the  process $\widehat{\mathbb{D}}(x)$. $\qed$

\bibliographystyle{apalike}
\bibliography{cglA}

\end{document}